\documentclass[11pt]{article}
\usepackage{ulem}
\usepackage{amsfonts,amsthm,amssymb,amsmath,mathrsfs}
\usepackage{graphicx,color,hyperref,subfig,mathtools}
\usepackage[usenames,dvipsnames]{xcolor}
\usepackage{pdfsync}
\usepackage[margin=1.0in]{geometry}

\usepackage[T1]{fontenc}
\usepackage[scaled=.92]{helvet} [scaled=.92]

\usepackage[round]{natbib}
\usepackage{hyperref}
\usepackage{authblk}

\usepackage{multibib}
\newcites{A}{References for Appendices}

\usepackage{tabularx, booktabs, array, multirow,enumitem}

\usepackage{mdframed}

\usepackage{pgfplots}
\usepackage{tikz}
\usetikzlibrary{calc,shapes,arrows.meta,pgfplots.groupplots,decorations.pathreplacing} 

\newtheorem{theorem}{Theorem}

\newtheorem{lemma}{Lemma}

\newtheorem{proposition}{Proposition}

\usepackage{chngcntr}
\usepackage{apptools}
\AtAppendix{\counterwithin{figure}{section}}

\usepackage[english]{babel}
\usepackage[utf8]{inputenc}
\usepackage[colorinlistoftodos]{todonotes}
\usepackage{algorithm}
\usepackage{algpseudocode}

\usepackage{chngcntr}

\newcounter{parentnumber}


\usepackage[hang,flushmargin]{footmisc}

\allowdisplaybreaks

\usepackage{array}
\newcolumntype{P}[1]{>{\centering\arraybackslash}p{#1}}
\newcolumntype{M}[1]{>{\centering\arraybackslash}m{#1}}

\usepackage{amsmath}
\usepackage{amsfonts}
\usepackage{graphicx}
\usepackage{fancyvrb}
\usepackage{url}
\usepackage{subfig}
\usepackage{listings}
\usepackage{setspace}
\usepackage{tikz}
\usetikzlibrary{automata,positioning}
\usepackage{pgfplots}

\title{Mutation Order and Selection Shape Intratumor Heterogeneity in Tumor Evolution}

\author{Yichen Chu$^{1}$ \quad \quad \hspace*{-6pt} Xiaoxia Sheng$^{1}$ \quad  \hspace*{-6pt}Xuanming Zhang$^{2,*}$ \quad \hspace*{-6pt}Zicheng Wang$^{1,*}$}
\date{%
    {\scriptsize $^1$School of Data Science, The Chinese University of Hong Kong, Shenzhen, Guangdong 518172, China \\ $^2$Department of Industrial and Systems Engineering, University of Minnesota, Twin Cities, U.S. \\
    $^{*}$Correspondence:  zhan8093@umn.edu and wangzicheng@cuhk.edu.cn}
}

\begin{document}

\maketitle

\begin{abstract}
\singlespacing 
Cancer progression often requires the accumulation of multiple driver mutations, yet the same set of drivers may be acquired in different orders. How these alternative mutation-order pathways jointly shape tumor clonal structure remains unclear. We develop a multitype branching-process model in which malignant transformation requires two driver mutations and distinguish malignant cells both by their mutation order and by the independent transformation event that founded their clone. Under a successive exponential approximation, we establish point-process limits for the pathway-specific clone-size processes and derive a closed-form expression for the limiting expected Simpson's index of the combined malignant population. When the two mutation orders produce malignant cells with the same net growth rate, the limiting index separates into effective pathway weights, determined by mutation rates and birth--death dynamics at preceding stages, and within-pathway concentration terms, determined by intermediate-to-malignant growth-rate ratios. This decomposition shows that a driver's effect on heterogeneity depends critically on when it is acquired. A strong driver acquired early expands the intermediate lineage and increases the supply of independent malignant founders, whereas the same driver acquired at the final transition strengthens the growth and age advantage of early-founded malignant clones. Under additive fitness effects, these counteracting mechanisms can produce a non-monotone relationship between selective advantage and clonal concentration. We further show that threshold-like non-additive fitness effects can generate highly concentrated malignant populations, while order-dependent terminal fitness leads the faster-growing pathway to dominate asymptotically. Together, these results reveal how mutation order, mutational accessibility, selection, and epistasis jointly determine lineage-level intratumor heterogeneity.

\singlespacing 
\end{abstract}
\thispagestyle{empty}

\clearpage

\section{Introduction}

Cancer progression is widely viewed as a multistage evolutionary process driven by the accumulation of driver mutations. These mutations can affect oncogenes, tumor suppressor genes, DNA-repair pathways, and other regulators of cell growth and survival. A driver mutation that gives a cell a selective advantage can promote clonal expansion. Cells within the expanding lineage may then acquire further driver mutations and eventually become malignant. This evolutionary view of cancer, first formalized through the clonal evolution model, emphasizes that tumor development is not merely a sequence of molecular events, but a dynamic process shaped by mutation, selection, and population growth (\citealt{nowell1976clonal, vogelstein2004cancer}).

The order in which driver mutations are acquired need not be fixed. Classical models often represent cancer progression as a largely linear process in which cells pass through a prescribed sequence of intermediate states before becoming malignant (\citealt{davis2017tumor}). However, genomic evidence indicates that many tumors follow branching rather than strictly linear evolutionary trajectories. Multiregion sequencing has identified genetically distinct subclonal lineages in spatially separated regions of the same tumor, indicating that multiple cell populations can evolve in parallel within a single neoplasm (\citealt{gerlinger2012intratumor}). When multiple mutational routes are biologically accessible, malignant cells may therefore arise through distinct mutation-order pathways, even when those pathways involve the same set of driver mutations (\citealt{nicholson2019competing,paterson2020mathematical}). Clinical observations provide further support for this possibility. In myeloproliferative neoplasms, for example, JAK2 and TET2 mutations can be acquired in either order across patients, giving rise to distinct mutational trajectories (\citealt{ortmann2015effect}).

Because branching evolution allows malignant cells to arise through multiple mutational routes, it generates a more complex clonal structure than that generated by a linear progression model (\citealt{davis2017tumor}). Specifically, driver mutations may confer different selective advantages, such that their order of acquisition affects the growth and relative abundance of intermediate lineages and, in turn, subsequent progression to malignancy. Driver mutations may also occur at different rates, making some routes more accessible than others. Moreover, both selective effects and mutation rates may depend on the current genetic state. Epistatic interactions can cause the fitness effect of a mutation to vary with the genetic background in which it occurs, whereas state-dependent mutation rates can modulate the accessibility of each pathway. Together, pathway-specific differences in intermediate-lineage growth, mutational accessibility, and malignant-cell fitness can shape the relative contributions of different pathways to malignant progression, as well as the number and size distribution of malignant clones generated through each pathway.

The clonal structure generated by these pathway-specific evolutionary dynamics constitutes an important component of intratumor heterogeneity, which more broadly encompasses genetic, epigenetic, and phenotypic variation among cell populations within the same tumor. In this paper, we focus on this \emph{lineage-level intratumor heterogeneity} (\citealt{mcdonald2015multitype}). Specifically, we define a malignant clone genealogically as all malignant descendants of a single cell that \sout{first} enters the malignant state, and we treat lineages founded by distinct malignant-transformation events as different clones even when their founders carry the same set of driver mutations. This clonal structure is clinically relevant because independently evolving lineages can accumulate distinct genetic or phenotypic changes and may therefore differ in their sensitivity to treatment. Depending on the tumor type and therapeutic intervention, resistant subclones may already be present before treatment or may emerge during therapy, allowing part of the tumor population to survive and subsequently contribute to relapse (\citealt{dagogo2018tumour}). Clonal genetic diversity has also been associated with disease progression in premalignant conditions such as Barrett's esophagus, where greater diversity predicts a higher risk of progression to esophageal adenocarcinoma (\citealt{maley2006genetic}). Despite these clinical implications, how pathway-specific evolutionary dynamics determine the number and relative sizes of independently founded malignant clones, and hence this lineage-level intratumor heterogeneity, remains insufficiently understood.

In this paper, we develop a multitype branching-process model to investigate how pathway-specific evolutionary dynamics shape malignant progression and intratumor heterogeneity. This framework is well suited to our purpose because differences in the selective advantages conferred by driver mutations can be represented through type-specific growth rates, differences in mutational accessibility through type-specific mutation rates, and epistatic interactions by allowing these rates to depend on the current genetic state. We consider a population initiated by a single wild-type cell that must acquire two driver mutations to reach malignancy.\footnote{The two-driver setting is adopted for analytical clarity and to isolate the role of alternative mutation orders. The same multitype branching-process framework can be formulated for progression requiring more than two driver mutations. The qualitative insights developed below carry over to such higher-dimensional mutation networks.} The two possible orders of acquisition define distinct evolutionary pathways. Cells that acquire mutation 1 before mutation 2 are distinguished from those that acquire mutation 2 before mutation 1, even though both ultimately carry the same pair of modeled driver mutations. Our analysis characterizes the relative representation of the two pathways in the malignant population and the resulting clonal concentration, quantified by Simpson's index. Together, these quantities show how pathway-specific growth and mutation dynamics determine both pathway composition and the distribution of malignant cells among independently founded clones.

Building on this framework, we characterize the biological mechanisms through which alternative mutation-order pathways shape lineage-level intratumor heterogeneity and develop the theoretical analysis needed to quantify these effects. Biologically, we show that the intratumor heterogeneity generated by alternative mutation-order pathways is jointly determined by the relative representation of those pathways and by the number and size distribution of independently founded malignant clones within each pathway. Our analysis further shows that the effect of a driver's selective advantage depends on the stage at which it is acquired. A larger fitness gain at an intermediate stage can accelerate the expansion of the intermediate population and increase the rate at which new malignant clones are founded. By contrast, a larger fitness gain at the final transition to malignancy can widen the growth-rate difference between malignant clones and the intermediate population that produces them, strengthening the age advantage of early-founded clones and increasing clonal concentration. Theoretically, we extend point-process analyses of intratumor heterogeneity from a single linear mutational sequence (\citealt{durrett2011intratumor}) to a branched mutation network. Under the successive exponential approximation, we establish vague convergence of the finite-time pathway-specific clone-size point processes and show that the associated expected Simpson's indices converge to their point-process limits. We then derive an explicit expression for the limiting expected Simpson's index when the malignant populations arising from both mutation-order pathways have equal growth rates. This expression accounts for the dependence between within-pathway clonal concentration and random pathway representation and separates the pathway-specific parameters that determine effective pathway weights from the intermediate-to-malignant growth-rate ratios that determine within-pathway clonal concentration.

The remainder of this paper is organized as follows. Section~\ref{Sec:Literature} reviews the related literature on multistage cancer progression, intratumor heterogeneity, and competing mutation-order pathways. Section~\ref{Sec:Formulation} introduces the multitype branching-process model and defines pathway representation and lineage-level heterogeneity through Simpson's index, which gives the probability that two malignant cells sampled independently and uniformly, with replacement, descend from the same malignant founder. Smaller values indicate greater lineage-level clonal diversity. We adopt this measure for its direct probabilistic interpretation and because its expression as a sum of squared clone frequencies facilitates explicit analysis within our branching-process framework. Section~\ref{Sec:Results} presents the main results. We first analyze the deterministic-order benchmark, then characterize the limiting expected Simpson's index when the two mutation orders coexist under commutative fitness, and finally examine additive, commutative but non-additive, and non-commutative fitness effects. Section~\ref{Sec:conclusion} concludes with a discussion of the biological implications, limitations, and directions for future research. The appendices contain the proofs and additional technical details.

\section{Related Literature}\label{Sec:Literature} 

Our work relates to three streams of literature: stochastic models of multistage cancer progression, empirical and theoretical studies of intratumor heterogeneity, and research on competing mutational pathways and alternative orders of driver-mutation acquisition.

The first stream develops mathematical models of mutation accumulation during carcinogenesis. Classical multistage theories used patterns of age-specific cancer incidence and hereditary cancer data to infer that malignant transformation requires the accumulation of multiple genetic events (\citealt{armitage1954age,knudson1971mutation}). Subsequent stochastic models, particularly multitype branching processes, provide a more mechanistic description of how mutant lineages arise, expand, and acquire additional alterations. These models have been used to study the waiting time to malignancy, the emergence and expansion of successive mutant types, the accumulation of driver and passenger mutations, and the distribution of premalignant lesions (\citealt{iwasa2006evolution,beerenwinkel2007genetic,haeno2007evolution,durrett2009waiting,durrett2010evolution,bozic2010accumulation,dewanji2011number,nicholson2023sequential,zhang2023waiting,zhang2024accumulation}). Related branching-process models have also been developed to analyze cancer recurrence and therapy resistance (\citealt{avanzini2019cancer,hanagal2022large,li2023comparison,leder2024parameter,leder2025parameter,leder2026clonal}). Our model builds on this literature by using a multitype birth--death process to describe the stochastic emergence and expansion of mutant lineages. In contrast to studies that focus primarily on the waiting time to malignancy or the total mutant population, we examine how pathway-specific growth and mutation dynamics determine the relative representation of alternative mutation-order pathways in the malignant population and the clonal diversity of that population.

The second stream examines the intratumor heterogeneity of evolving tumors. Empirical studies have shown that tumors often contain multiple genetically distinct subclones and may evolve along branching rather than strictly linear trajectories (\citealt{gerlinger2012intratumor,davis2017tumor}). Such heterogeneity has important clinical implications because subclonal variation can influence disease progression, treatment response, and the likelihood of relapse. For example, clonal diversity in Barrett's esophagus has been associated with an increased risk of progression to esophageal adenocarcinoma (\citealt{maley2006genetic}). On the theoretical side, branching-process models have been used to characterize clone-size distributions and genetic diversity in growing tumors. In particular, \citealt{durrett2011intratumor} analyzed intratumor heterogeneity in evolutionary models of tumor progression using quantities such as Simpson's index and the size of the largest clone. \citealt{storey2017spatial} derived estimates of Simpson's index for the premalignant population under a two-step carcinogenesis model and introduced spatial measures characterizing the typical scale of genetic heterogeneity and the extent of a premalignant clone surrounding a point biopsy. \citealt{cheek2020genetic} characterized mutation frequencies and the genetic composition of exponentially growing cell populations, including settings that relax the infinite-sites assumption and allow for cell death, selection, and  site-specific mutation rates. Other studies have derived asymptotic clone-size distributions and site-frequency spectra under different assumptions about tumor growth (\citealt{nicholson2016universal,cheek2018mutation,stein2025patterns,gunnarsson2021exact,ahmed2026site}). Our work is closest in spirit to this literature but differs in the source and organization of the intratumor heterogeneity that it studies. We distinguish malignant cells both by the mutation-order pathway through which they arise and by the independent transformation event that founds each clone. This lineage-based representation allows us to separate the relative concentration of malignant cells across evolutionary pathways from the diversity of independently founded clones within each pathway and to determine how both contribute to the resulting intratumor heterogeneity.

The third stream studies competing mutational pathways and the biological consequences associated with alternative mutation orders. Empirical studies have shown that the order in which driver mutations are acquired may be associated with differences in tumor phenotype, disease classification, and subsequent evolutionary dynamics. In myeloproliferative neoplasms, for example, the order of JAK2 and TET2 acquisition is associated with distinct molecular and clinical features (\citealt{ortmann2015effect}). Broader reviews have similarly emphasized that initiating mutations, cell of origin, and mutation order can jointly influence cancer development (\citealt{kent2017order,levine2019roles}). Mathematical studies have primarily examined which mutational sequences are realized and how quickly particular evolutionary endpoints are reached. \citealt{nicholson2019competing} studied competing mutational pathways in growing populations, with applications to multi-drug resistance, and derived the distribution of mutational routes in regimes where resistant mutants are rare. \citealt{paterson2020mathematical} developed a model of colorectal cancer initiation involving multiple alterations and showed that the distribution of mutation sequences can depend strongly on the selective advantages of intermediate genotypes. Related work has examined the temporal ordering of mutations during cancer initiation (\citealt{teimouri2021temporal}), the order of driver mutations in colorectal cancer (\citealt{li2023mathematical}), and mutation-order effects in myeloproliferative neoplasms (\citealt{wang2024order}). Collectively, these studies show that pathway-specific fitness and mutation dynamics can influence the relative frequencies and timing of alternative mutation orders, while different orders may be associated with distinct biological outcomes. However, less is known about how competing mutation-order pathways jointly determine the clonal structure of the malignant population. We address this gap by characterizing how pathway-specific differences in intermediate-lineage growth, mutational accessibility, and malignant-cell fitness jointly determine the intratumor heterogeneity.

\section{Model}\label{Sec:Formulation}

We formulate cancer initiation and progression as a continuous-time multitype branching process initiated by a single wild-type cell. Cells divide, die, and may acquire driver mutations during division, thereby generating new cell types. To make the role of mutation order explicit, we consider the minimal setting in which malignant transformation requires the acquisition of two distinct driver mutations. This is the simplest setting that allows distinct mutational routes to converge on the same malignant driver genotype. Although the analysis below focuses on two driver mutations, the framework can be extended to settings in which malignant progression requires more than two driver mutations.

The process has five cell types, indexed by $\mathcal{S}=\{s_{\emptyset},s_1,s_2,s_{12},s_{21}\}$. The type \(s_{\emptyset}\) represents wild-type cells, whereas \(s_1\) and \(s_2\) represent cells carrying only mutation 1 and only mutation 2, respectively. Both \(s_{12}\) and \(s_{21}\) represent malignant cells carrying the two driver mutations, where the subscripts record the order in which the mutations were acquired. Specifically, cells of type \(s_{12}\) acquired mutation 1 before mutation 2, whereas cells of type \(s_{21}\) acquired mutation 2 before mutation 1. Thus, although these two types have the same modeled driver genotype, they correspond to distinct evolutionary pathways.

For each \(x\in\mathcal S\), let \(Z_x(t)\) denote the number of cells of type \(x\) at time \(t\). The process is initialized by a single wild-type cell, i.e., $Z_{s_{\emptyset}}(0)=1$ and $Z_{s_1}(0)=Z_{s_2}(0)=Z_{s_{12}}(0)=Z_{s_{21}}(0)=0$. A cell of type \(x\) undergoes symmetric division, producing two cells of type \(x\), at rate \(a_x\), and dies at rate \(b_x\). The corresponding net growth rate is $\lambda_x=a_x-b_x$, which we assume to be positive for every \(x\in\mathcal S\). Mutations are modeled as division events independent of the birth and death events that produce one daughter cell in the parental state and one daughter cell in a mutant state. Specifically, a cell in state \(x\) can divide into one cell of state \(x\) and one cell of state \(y\), with \(x\neq y\). The admissible transitions depend on the model under consideration. In the full two-driver model, the allowed mutation transitions are $s_{\emptyset}\to s_1$, $s_{\emptyset}\to s_2$, $s_1\to s_{12}$, and $s_2\to s_{21}$ with corresponding rates $\mu_{\emptyset,1}$, $\mu_{\emptyset,2}$, $\mu_{1,12}$, and $\mu_{2,21}$. All other mutation rates are set to zero unless stated otherwise.

\textbf{Deterministic order model.} 
As a benchmark, we first consider a linear, stepwise progression model in which mutation 1 must be acquired before mutation 2. Specifically, we assume $\mu_{\emptyset,1}>0$ and $\mu_{1,12}>0$, and set all other mutation rates to zero. The only admissible mutational pathway is therefore
\[
s_{\emptyset}\longrightarrow s_1\longrightarrow s_{12}.
\]
Thus, every malignant cell belongs to type \(s_{12}\) and arises through the same mutational sequence. The term deterministic refers here to the prescribed order of mutation acquisition; the timing of mutation and the subsequent population dynamics remain stochastic. This model provides a benchmark for examining how the relative fitness effects of mutations acquired at different stages of progression influence malignant intratumor heterogeneity. In particular, we compare the cases $\lambda_{s_{12}}-\lambda_{s_1}>\lambda_{s_1}-\lambda_{s_{\emptyset}}$ and $\lambda_{s_{12}}-\lambda_{s_1}<\lambda_{s_1}-\lambda_{s_{\emptyset}}$. In the first case, acquisition of the second mutation confers a larger increase in net growth rate than acquisition of the first mutation. In the second case, acquisition of the first mutation confers the larger increase.

\textbf{Random order model.} 
We next allow the two driver mutations to be acquired in either order. The nonzero mutation rates are $\mu_{\emptyset,1}>0$, $\mu_{\emptyset,2}>0$, $\mu_{1,12}>0$, $\mu_{2,21}>0$, giving rise to two admissible mutational pathways:
\[
s_{\emptyset}\longrightarrow s_1\longrightarrow s_{12},
\qquad
s_{\emptyset}\longrightarrow s_2\longrightarrow s_{21}.
\]
Thus, malignant cells may arise through either mutational sequence, and the pathway followed by any given lineage is determined stochastically. For technical reasons, we assume that the acquisition of an additional driver mutation increases the net growth rate along either pathway, i.e., $\lambda_{s_{\emptyset}} < \lambda_{s_{1}} < \lambda_{s_{12}}$ and $\lambda_{s_{\emptyset}} < \lambda_{s_{2}} < \lambda_{s_{21}}$.

\textbf{Intratumor heterogeneity.}
In the random-order models, malignant cells can arise through either the \((1,2)\) pathway or the \((2,1)\) pathway. For notational convenience, we write $Z_{12}(t)=Z_{s_{12}}(t)$ and $Z_{21}(t)=Z_{s_{21}}(t)$. We use $X_i(t)$ to denote the number of cells at time $t$ in the $i$-th malignant clone arising along the (1, 2) pathway, and $Y_j(t)$ to denote the number of cells at time $t$ in the $j$-th malignant clone arising along the (2, 1) pathway. The total number of malignant cells at time \(t\) is $M(t)=Z_{12}(t)+Z_{21}(t)$. Conditional on the event \(M(t)>0\), define the pathway fractions $\pi_{12}(t)=\frac{Z_{12}(t)}{M(t)}$ and $\pi_{21}(t)=\frac{Z_{21}(t)}{M(t)}$. These quantities satisfy \(\pi_{12}(t)+\pi_{21}(t)=1\) and measure the representation of the two mutation-order pathways within the malignant population. In the deterministic-order model, $Z_{21}(t)=0$ for all \(t\), and therefore \(\pi_{12}(t)=1\) whenever \(M(t)>0\).

\[
R_t
=
\frac{
\sum_i X_i(t)^2+\sum_j Y_j(t)^2
}{
M(t)^2
}
=
\frac{
\sum_i X_i(t)^2+\sum_j Y_j(t)^2
}{
\left(Z_{12}(t)+Z_{21}(t)\right)^2
}.
\]
Equivalently, \(R_t\) is the probability that two malignant cells sampled independently and uniformly, with replacement, from the population at time \(t\) belong to the same clone. A larger value of \(R_t\) indicates greater clonal concentration, whereas a smaller value indicates greater clonal diversity.

We also define the within-pathway Simpson's indices $R_{(12),t}
=\sum_i\left(\frac{X_i(t)}{Z_{12}(t)}\right)^2$ when \(Z_{12}(t)>0\), and $R_{(21),t}=\sum_j\left(\frac{Y_j(t)}{Z_{21}(t)}\right)^2$ when \(Z_{21}(t)>0\). When a pathway contains no malignant cells at time \(t\), its within-pathway index is set to zero by convention. The Simpson's index of the total malignant population then admits the decomposition
\begin{align*}
    R_t& =
R_{(12),t}
\left(
\frac{Z_{12}(t)}
{Z_{12}(t)+Z_{21}(t)}
\right)^2
+
R_{(21),t}
\left(
\frac{Z_{21}(t)}
{Z_{12}(t)+Z_{21}(t)}
\right)^2\\
& = R_{(12),t}\pi_{12}(t)^2+R_{(21),t}\pi_{21}(t)^2.
\end{align*}
The squared pathway fractions describe how malignant cells are distributed across the two evolutionary pathways, whereas the within-pathway indices describe how cells are distributed among independently founded clones within each pathway. This decomposition therefore identifies two channels through which pathway-specific evolutionary dynamics shape the malignant intratumor heterogeneity: they determine the relative representation of the two pathways and the number and relative sizes of malignant clones arising within each pathway.

\section{Results}\label{Sec:Results}

In this section, we present our main results. Since Simpson's index depends on the full collection of malignant clone sizes, its exact analysis is generally difficult. We therefore study the deterministic-order and random-order models introduced in Section~\ref{Sec:Formulation} using the \emph{successive exponential approximation} developed by \citealt{durrett2010evolution} and \citealt{durrett2011intratumor}. The approximation replaces the population producing each new mutational wave by an exponential growth trajectory with a random amplitude. This amplitude retains the effects of early stochastic fluctuations, while the exponential trajectory simplifies the subsequent production of mutant clones.

Conditional on nonextinction of the wild-type lineage, we approximate the wild-type population by
\[
Z^*_{s_{\emptyset}}(t)
=
V_0e^{\lambda_{s_{\emptyset}}t},
\qquad t\in\mathbb{R},
\]
where \(V_0\) is exponentially distributed with rate
\(\lambda_{s_{\emptyset}}/a_{s_{\emptyset}}\). This representation is motivated by the long-time exponential growth of a surviving birth--death process, with \(V_0\) capturing stochastic variation during its early expansion. Conditional on \(V_0\), the approximating wild-type population follows a deterministic exponential trajectory and produces single-mutant founders through independent Poisson processes with intensities
\(\mu_{\emptyset,1}V_0e^{\lambda_{s_{\emptyset}}t}\) and
\(\mu_{\emptyset,2}V_0e^{\lambda_{s_{\emptyset}}t}\). Each founder then initiates an independent birth--death lineage. The approximation is applied recursively by aggregating the long-time contributions of these lineages and using the resulting random amplitudes to represent the single-mutant populations as
\[
Z^*_{s_1}(t)
=
V_{s_1}e^{\lambda_{s_1}t},
\qquad
Z^*_{s_2}(t)
=
V_{s_2}e^{\lambda_{s_2}t},
\qquad t\in\mathbb{R}.
\]
Here, \(V_{s_1}\) and \(V_{s_2}\) incorporate randomness in mutation times, lineage survival, and early stochastic expansion. Their construction and distributions are derived in the Appendix. These approximating populations, in turn, generate malignant founders at rates
\(\mu_{1,12}V_{s_1}e^{\lambda_{s_1}t}\) and
\(\mu_{2,21}V_{s_2}e^{\lambda_{s_2}t}\), while individual malignant clones retain their stochastic birth--death dynamics.

This construction makes the analysis tractable because, conditional on the source amplitudes, clone-founding times form Poisson processes with explicit exponential intensities. Extending the exponential trajectories to all \(t\in\mathbb{R}\) is an additional idealization that removes the finite starting-time boundary. Together, these steps yield explicit Poisson point-process representations of the limiting rescaled clone sizes, whose Laplace transforms allow us to evaluate the limiting expected Simpson's index. Unless otherwise stated, all asymptotic results in this section, including the clone size $X_i(t)$ and $Y_j(t)$, and the limiting Simpson's indices, are derived under this successive exponential approximation. For notational simplicity, we suppress the superscript \(*\) in what follows.

\subsection{Deterministic order model}

We first analyze the deterministic-order model, in which mutation 1 must be acquired before mutation 2. All malignant cells therefore arise through the \((1,2)\) pathway and are of type \(s_{12}\). Recall that \(X_i(t)\) denotes the size of the \(i\)-th malignant clone at time \(t\). The total malignant population is $Z_{12}(t)=\sum_i X_i(t)$, and, conditional on \(Z_{12}(t)>0\), its Simpson's index is $R_t=\sum_i\left(\frac{X_i(t)}{Z_{12}(t)}\right)^2$. By convention, we set \(R_t=0\) when \(Z_{12}(t)=0\).

Under the successive exponential approximation, the \(s_1\) population  grows as \(V_{s_1}e^{\lambda_{s_1}t}\) and seeds malignant \(s_{12}\)  clones at an exponentially increasing rate, while each established  \(s_{12}\) clone grows at rate \(\lambda_{s_{12}}\).  Lemma~\ref{lem:one_step_clone_process} in the Appendix, applied to the  transition \(s_1\to s_{12}\), establishes that the rescaled finite-time  \(s_{12}\)-clone-size process converges vaguely to its limiting Poisson  point process. The single-wave result of  \citealt{durrett2011intratumor} gives the expected Simpson's index of this limiting point process. The one-pathway specialization of Lemma~\ref{lem:simpson_convergence} in the Appendix then shows that the expected Simpson's index of the finite-time population converges to this limiting expectation. Together, these results yield the following proposition.
\begin{proposition}
\label{lem:deterministic_simpson}
Suppose that $0<\lambda_{s_{\emptyset}}<\lambda_{s_1}<\lambda_{s_{12}}$. Under the successive exponential approximation,
\[
\lim_{t\to\infty}
\mathbb{E}\!\left[
R_t
\right]
=
1-\lambda_{s_1}/\lambda_{s_{12}}.
\]
\end{proposition}

To express this result in terms of the fitness gains conferred by the two mutations, define $f_1=\lambda_{s_1}-\lambda_{s_{\emptyset}}$ and $f_2=\lambda_{s_{12}}-\lambda_{s_1}$. Here, \(f_1\) is the increase in net growth rate following acquisition of the first mutation, while \(f_2\) is the increase following acquisition of the second mutation. Since $\lambda_{s_{12}}=\lambda_{s_{\emptyset}}+f_1+f_2$, Proposition~\ref{lem:deterministic_simpson} gives
\begin{equation}
\label{eqn:limiting_simpson_deterministic}
\lim_{t\to\infty}
\mathbb{E}\!\left[
R_t
\right]
=
\frac{f_2}
{\lambda_{s_{\emptyset}}+f_1+f_2}.
\end{equation}
Equation~\eqref{eqn:limiting_simpson_deterministic} identifies two opposing within-pathway effects of stage-specific selection. Indeed,
\[
    \frac{\partial}{\partial f_1}
    \left(
    \frac{f_2}
    {\lambda_{s_{\emptyset}}+f_1+f_2}
    \right)
    =
    -\frac{f_2}
    {\left(\lambda_{s_{\emptyset}}+f_1+f_2\right)^2}
    <0,
\]
whereas
\[
    \frac{\partial}{\partial f_2}
    \left(
    \frac{f_2}
    {\lambda_{s_{\emptyset}}+f_1+f_2}
    \right)
    =
    \frac{\lambda_{s_{\emptyset}}+f_1}
    {\left(\lambda_{s_{\emptyset}}+f_1+f_2\right)^2}
    >0.
\]
We refer to the effect of changing $f_1$ as the \emph{early-selection diversification effect}. A larger first-stage fitness gain \(f_1\) accelerates the expansion of the intermediate population and causes the supply of new malignant founders to increase more rapidly. Later-founded malignant clones can therefore retain greater representation, leading to a more even clone-size distribution and a smaller Simpson's index. We refer to the effect of changing $f_2$ as the \emph{late-selection concentration effect}. A larger second-stage fitness gain \(f_2\) widens the growth-rate difference between established malignant clones and the intermediate population that produces them. This strengthens the growth and age advantage of early-founded malignant clones and increases the Simpson's index. Here, ``early'' and ``late'' refer to the position of selection in the mutational sequence rather than to calendar time. These two effects arise within a single mutational pathway. They show that the conventional selective-sweep intuition that a stronger driver necessarily reduces diversity (\citealt{smith1974hitch}) is incomplete: selection at the final transition promotes clonal concentration, whereas selection at an earlier stage can promote diversity by expanding the population from which independent malignant clones are founded.

Proposition~\ref{lem:deterministic_simpson} also shows that, although the mutation rate controls the scale of the malignant-clone founding process, it does not enter the limiting Simpson's index. A higher mutation rate increases the rate at which malignant clones are founded, so malignant clones tend to arise earlier and more clones are present at any fixed time. However, it does not change the rate at which the founding intensity grows relative to the growth rate of established malignant clones. Consequently, the limiting index depends on the relative growth rates of the intermediate and malignant populations, \(\lambda_{s_1}/\lambda_{s_{12}}\), rather than on the mutation rate. The mutation rate nevertheless affects the timing of malignant progression and the number of malignant clones observed over finite time horizons.

\begin{figure}[h]
    \centering
    \includegraphics[width=0.75\textwidth]
    {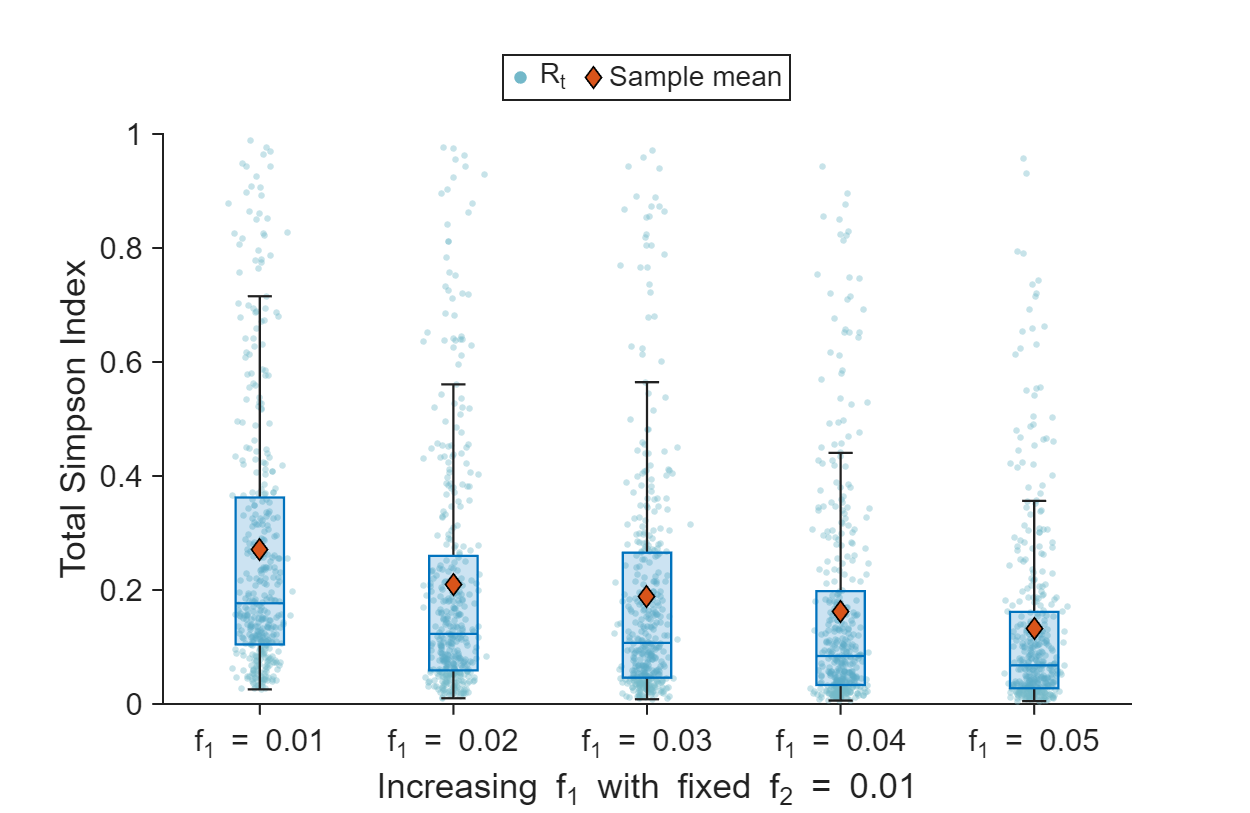}
    \caption{
    Finite-time distributions of the Simpson's index of the malignant population under the deterministic-order model. The second-stage fitness gain is fixed at \(f_2=0.01\), whereas the first-stage fitness gain varies over \(f_1\in\{0.01,0.02,0.03,0.04,0.05\}\). The death rate is set to \(0\), and both admissible mutation rates are set to \(\mu=10^{-3}\). Each boxplot summarizes 400 independent simulation realizations. The horizontal line inside each box represents the sample median, while the lower and upper bounds of the box denote the 25th and 75th percentiles, respectively. Red diamonds indicate sample means, and individual jittered points represent single simulation realizations.
}
    \label{fig:deterministic_order_clone_simpson}
\end{figure}

We simulate the deterministic-order model using the Gillespie algorithm. Each realization is initialized with a single wild-type cell and stopped when the total cell population first reaches \(10^8\). Figure~\ref{fig:deterministic_order_clone_simpson} provides a finite-time illustration of the early-selection diversification effect. With \(f_2\) fixed, equation~\eqref{eqn:limiting_simpson_deterministic} implies that the limiting expected Simpson's index is strictly decreasing in \(f_1\). A larger \(f_1\) accelerates the expansion of the intermediate \(s_1\) population, causing the founding intensity of new malignant clones to 
increase more rapidly relative to the growth of established clones and thereby allowing later-founded clones to retain greater representation. Consistent with this prediction, as \(f_1\) increases from \(0.01\) to \(0.05\), the sample mean of the Simpson's index decreases from \(0.270\) to \(0.133\).

\subsection{Random Order Model}

We next consider the random-order model, in which malignant cells may arise through either the \(s_{\emptyset}\to s_1\to s_{12}\) pathway or the \(s_{\emptyset}\to s_2\to s_{21}\) pathway. We first focus on the case in which the two mutation orders produce malignant cell types with the same net growth rate, i.e., $\lambda_{s_{12}}=\lambda_{s_{21}}=\lambda$. We refer to this setting as the \emph{commutative-fitness regime}, because the net growth rate of a malignant cell depends on the set of acquired mutations but not on their order of acquisition. In this regime, both malignant populations grow on the same exponential scale. Neither pathway becomes asymptotically negligible, and their relative contributions are instead determined by the mutation and growth dynamics at the preceding stages of progression.

Recall that the Simpson's index of the total malignant population can be decomposed as $R_t=R_{(12),t}\pi_{12}(t)^2+R_{(21),t}\pi_{21}(t)^2$. This decomposition separates the clonal diversity within each pathway from the representation of that pathway in the malignant population. Although Proposition~\ref{lem:deterministic_simpson} gives the limiting within-pathway index when either pathway is considered in isolation, these two single-pathway limits cannot be substituted directly into the decomposition. For each pathway, the within-pathway index and the corresponding pathway fraction depend on the same malignant clone sizes and are therefore dependent. In addition, the two pathway fractions are coupled through the total malignant population in their common denominator. The following theorem accounts for these  dependencies and derives an exact expression for the limiting expected Simpson's index of the combined malignant population.

\begin{theorem}\label{thm:expected_simpson_additive}
For any admissible mutation \(s_i\to s_j\) with
\(\lambda_{s_i}<\lambda_{s_j}\), define
\begin{equation}
\label{eq:H_constant}
H_{i,j}
=
\frac{1}{a_{s_j}}
\left(
\frac{a_{s_j}}{\lambda_{s_j}}
\right)^{\lambda_{s_i}/\lambda_{s_j}}
\frac{\pi}{\sin\left(\pi\lambda_{s_i}/\lambda_{s_j}\right)} .
\end{equation}
Let
\begin{equation}
\label{eq:A_constants}
A_{12}
=
\mu_{\emptyset,1}H_{\emptyset,1}
\left(
\mu_{1,12}H_{1,12}
\right)^{\lambda_{s_{\emptyset}}/\lambda_{s_1}} \text{ and } A_{21}=\mu_{\emptyset,2}H_{\emptyset,2}\left(
\mu_{2,21}H_{2,21}\right)^{\lambda_{s_{\emptyset}}/\lambda_{s_2}} .
\end{equation}
Then, under the successive exponential approximation,
\begin{equation}
\label{eq:expected_simpson_additive}
\lim_{t\to\infty}\mathbb E[R_t]
=
\frac{
A_{12}\left(1-\frac{\lambda_{s_1}}{\lambda}\right)
+
A_{21}\left(1-\frac{\lambda_{s_2}}{\lambda}\right)
}{
A_{12}+A_{21}
}.
\end{equation}
\end{theorem}

By Theorem \ref{thm:expected_simpson_additive}, the limiting expected Simpson's index is a weighted average of the two single-pathway limits, $1-\frac{\lambda_{s_1}}{\lambda}$ and $1-\frac{\lambda_{s_2}}{\lambda}$. The proof of Theorem~\ref{thm:expected_simpson_additive} more specifically identifies the contribution of each pathway:
\begin{align}
\lim_{t\to\infty}
\mathbb{E}\!\left[
R_{(12),t}\pi_{12}(t)^2
\right]
&=
\frac{A_{12}}{A_{12}+A_{21}}
\left(
1-\frac{\lambda_{s_1}}{\lambda}
\right),
\label{eq:pathway_12_simpson_contribution}
\\
\lim_{t\to\infty}
\mathbb{E}\!\left[
R_{(21),t}\pi_{21}(t)^2
\right]
&=
\frac{A_{21}}{A_{12}+A_{21}}
\left(
1-\frac{\lambda_{s_2}}{\lambda}
\right).
\label{eq:pathway_21_simpson_contribution}
\end{align}
Equations~\eqref{eq:pathway_12_simpson_contribution} and \eqref{eq:pathway_21_simpson_contribution} show that the two within-pathway effects identified in the deterministic-order benchmark continue to govern clonal concentration along each mutation order. The terms $1-\frac{\lambda_{s_1}}{\lambda}$ and $1-\frac{\lambda_{s_2}}{\lambda}$ are the single-pathway limiting indices associated with the \(s_1\to s_{12}\) and \(s_2\to s_{21}\) transitions, respectively. The new feature of the random-order model is that these two within-pathway indices enter the total index with the effective weights $\frac{A_{12}}{A_{12}+A_{21}}$ and $\frac{A_{21}}{A_{12}+A_{21}}$. Changes in mutation rates, growth rates, or birth--death dynamics can therefore alter the total Simpson's index by shifting effective representation between pathways with different within-pathway clonal structures. We refer to this additional channel as the \emph{pathway-weight effect}. It decreases the total index when effective weight shifts toward the pathway with the smaller within-pathway index and increases the total index when weight shifts toward the pathway with the larger within-pathway index.

The constants \(A_{12}\) and \(A_{21}\) determine the pathway-weight effect by combining the mutation rates and birth--death parameters associated with the two stages of each pathway. In particular, \(\mu_{\emptyset,1}\) and \(\mu_{\emptyset,2}\) determine the rates at which the two intermediate populations are founded, whereas \(\mu_{1,12}\) and \(\mu_{2,21}\) determine the rates at which those populations produce malignant founders. For each transition \(s_i\to s_j\), the factor \(H_{i,j}\) summarizes how the growth rate of the \(s_i\) population and the birth--death dynamics of newly founded \(s_j\) lineages determine the asymptotic scale of the resulting \(s_j\) population. The ratio \(\lambda_{s_i}/\lambda_{s_j}\) compares the exponential rate at which the \(s_i\) population, and hence the founding intensity of new \(s_j\) lineages, increases with the rate at which established \(s_j\) lineages grow. The birth rate \(a_{s_j}\), together with the net growth rate \(\lambda_{s_j}=a_{s_j}-b_{s_j}\), determines the survival probability and early stochastic expansion of a newly founded \(s_j\) lineage. Consequently, even cell types with the same net growth rate may contribute differently to the pathway weights if their birth--death turnover rates differ. The exponents \(\lambda_{s_{\emptyset}}/\lambda_{s_1}\) and \(\lambda_{s_{\emptyset}}/\lambda_{s_2}\) arise from applying the exponential approximation recursively across the two successive mutational waves. Thus, \(A_{12}\) and \(A_{21}\) jointly capture the founding and expansion of intermediate lineages and malignant clones along the two mutation-order pathways. The ratios
$\frac{A_{12}}{A_{12}+A_{21}}$ and $\frac{A_{21}}{A_{12}+A_{21}}$therefore serve as effective pathway weights in the limiting expected Simpson's index.

The appearance of mutation rates in these weights marks an important  difference from the deterministic-order benchmark. Along a single pathway, mutation rates do not appear in Proposition~\ref{lem:deterministic_simpson}. In the random-order model, however, pathway-specific mutation rates can change the relative scales of the two malignant populations and thereby shift their relative contributions to the  combined malignant population. Mutation rates therefore affect the overall limiting Simpson's index through the pathway weights, while the within-pathway indices \(1-\frac{\lambda_{s_1}}{\lambda}\) and \(1-\frac{\lambda_{s_2}}{\lambda}\) remain determined solely by the corresponding growth-rate ratios.

Because the two weights are positive and sum to one, the limiting expected Simpson's index lies between the two single-pathway limits:
\[
\min\left\{
1-\frac{\lambda_{s_1}}{\lambda},
1-\frac{\lambda_{s_2}}{\lambda}
\right\}
\leq
\lim_{t\to\infty}\mathbb E[R_t]
\leq
\max\left\{
1-\frac{\lambda_{s_1}}{\lambda},
1-\frac{\lambda_{s_2}}{\lambda}
\right\}.
\]
For example, suppose that \(\lambda_{s_1}>\lambda_{s_2}\). The \((1,2)\) pathway then has a smaller limiting within-pathway Simpson's index because its intermediate population produces malignant founders at a rate closer to the growth rate of established malignant clones. Holding the growth rates fixed, an increase in \(A_{12}\) relative to \(A_{21}\) places greater weight on this more diverse pathway and reduces the limiting Simpson's index. Conversely, an increase in the effective contribution of the \((2,1)\) pathway raises the limiting index. Thus, pathway-specific mutation and growth parameters affect overall clonal concentration through their interaction with the clone-size distribution generated within each pathway.

In the special case $\lambda_{s_1}=\lambda_{s_2}$, the two pathways have the same limiting within-pathway Simpson's index. Equation~\eqref{eq:expected_simpson_additive} then reduces to $\lim_{t\to\infty}\mathbb E[R_t]=1-\frac{\lambda_{s_1}}{\lambda}$, independently of $A_{12}$ and $A_{21}$, which is the same result as in the deterministic-order model. Mutation rates may still change the relative representation of the two pathways and the numbers of malignant clones observed at finite times. However, such changes do not affect the limiting expected Simpson's index when the intermediate populations have the same growth rate, consistent with our finding in the deterministic-order model that mutation rates do not affect this limit.

\subsection{Special Cases and Extensions of Theorem \ref{thm:expected_simpson_additive}}

\textbf{Special case I: additive fitness effects.}
We now apply Theorem~\ref{thm:expected_simpson_additive} to the random-order model under \emph{additive fitness effects}. With a slight abuse of notation, define $f_1:=\lambda_{s_1}-\lambda_{s_{\emptyset}}$ and $f_2:=\lambda_{s_2}-\lambda_{s_{\emptyset}}$, where \(f_1\) and \(f_2\) are the increases in net growth rate conferred by mutations 1 and 2, respectively. Additivity assumes that each mutation confers the same fitness gain regardless of whether it is acquired before or after the other mutation. Accordingly, $\lambda_{s_{12}}-\lambda_{s_1}=f_2$ and $\lambda_{s_{21}}-\lambda_{s_2}=f_1$. It follows that $\lambda_{s_{12}}=\lambda_{s_{21}}=\lambda_{s_{\emptyset}}+f_1+f_2$, so the additive-fitness model is a special case of the commutative-fitness regime. 

Additive fitness effects provide a biologically motivated benchmark when two driver mutations have approximately independent effects on cell growth. For example, additive proliferative effects of \textit{ORAOV1} and \textit{CCND1} co-expression have been reported in engineered squamous cell carcinoma models, suggesting that these two co-amplified oncogenes can promote cancer-cell growth through distinct mechanisms (\citealt{mahieu2024oraov1}). Although this evidence does not directly establish background-independent increments in net population growth, it motivates additivity as a useful approximation in selected biological settings. 

Substituting the definitions of $f_1$ and $f_2$ into Theorem~\ref{thm:expected_simpson_additive}
gives
\begin{equation}
\label{eq:additive_simpson_special}
\lim_{t\to\infty}\mathbb E[R_t]
=
\frac{
A_{12}\frac{f_2}{\lambda}
+
A_{21}\frac{f_1}{\lambda}
}
{A_{12}+A_{21}} .
\end{equation}
The additive-fitness model connects the two stage-specific effects identified in the deterministic-order benchmark through the identity of the driver mutation. Along the \((1,2)\) pathway, mutation \(1\) is acquired at the intermediate stage and mutation \(2\) is acquired at the final transition. Along the \((2,1)\) pathway, their positions are reversed. Consequently, an increase in \(f_1\) activates the early-selection diversification effect along the \((1,2)\) pathway and the late-selection concentration effect along the \((2,1)\) pathway. At the same time, the change in \(f_1\) may alter \(A_{12}\) and \(A_{21}\), generating a pathway-weight effect. Additivity therefore provides a particularly transparent setting in which the two within-pathway effects and the between-pathway effect can be studied jointly.

To examine these effects, fix \(f_2\), \(\lambda_{s_{\emptyset}}\), and all primitive model parameters not determined by \(f_1\), and define $\overline R(f_1)=\lim_{t\to\infty}\mathbb E[R_t]$, $\omega_{12}(f_1)=\frac{A_{12}}{A_{12}+A_{21}}$. Equation~\eqref{eq:additive_simpson_special} can then be written as $\overline R(f_1)=\omega_{12}(f_1)\frac{f_2}{\lambda}+\bigl(1-\omega_{12}(f_1)\bigr)\frac{f_1}{\lambda}$, where $\lambda=\lambda_{s_{\emptyset}}+f_1+f_2$. Differentiating with respect to \(f_1\) gives
\begin{equation}
\label{eq:additive_simpson_sensitivity}
\frac{d\overline R(f_1)}{df_1}
=
\underbrace{
\bigl(1-\omega_{12}(f_1)\bigr)
\frac{\lambda_{s_{\emptyset}}+f_2}{\lambda^2}
}_{\text{late-selection concentration effect}}
-
\underbrace{
\omega_{12}(f_1)\frac{f_2}{\lambda^2}
}_{\text{early-selection diversification effect}}
-
\underbrace{
\omega_{12}'(f_1)\frac{f_1-f_2}{\lambda}
}_{\text{pathway-weight effect}}.
\end{equation}
The first two terms in equation~\eqref{eq:additive_simpson_sensitivity} are the stage-specific within-pathway effects identified in the deterministic-order benchmark, now operating along different mutation orders. The positive term is the late-selection concentration effect along the \((2,1)\) pathway, where mutation \(1\) is acquired at the final transition to malignancy. Increasing \(f_1\) widens the growth-rate difference between the \(s_2\) intermediate population and the malignant \(s_{21}\) clones, strengthening the dominance of early-founded malignant clones. The second term is the early-selection diversification effect along the \((1,2)\) pathway, where mutation \(1\) is acquired at the intermediate stage. Increasing \(f_1\) accelerates the expansion of the \(s_1\) population and causes the supply of new malignant founders to grow more rapidly, thereby reducing clonal concentration within this pathway. The final term is the pathway-weight effect introduced by random mutation order. When \(f_1>f_2\), the \((1,2)\) pathway has the smaller within-pathway Simpson's index. An increase in \(\omega_{12}(f_1)\) therefore shifts greater effective weight toward the more diverse pathway and reduces the total index, whereas a decrease in \(\omega_{12}(f_1)\) has the opposite effect.

Suppose that the two pathways are symmetric under interchange of mutations 1 and 2. Then $\left.\omega_{12}(f_1)\right|_{f_1=f_2}=\frac{1}{2}$. Because the pathway-weight term vanishes at \(f_1=f_2\), equation~\eqref{eq:additive_simpson_sensitivity} gives
\[
\left.
\frac{d\overline R(f_1)}{df_1}
\right|_{f_1=f_2}
=
\frac{\lambda_{s_{\emptyset}}}
{2\bigl(\lambda_{s_{\emptyset}}+2f_2\bigr)^2}
>0.
\]
Thus, starting from equal fitness gains, a small increase in \(f_1\) raises the limiting expected Simpson's index. For larger \(f_1\), the derivative becomes negative whenever
\[
\omega_{12}(f_1)f_2
+
\lambda\omega_{12}'(f_1)(f_1-f_2)
>
\bigl(1-\omega_{12}(f_1)\bigr)
\bigl(\lambda_{s_{\emptyset}}+f_2\bigr).
\]
Hence, when the late-selection concentration effect dominates near \(f_1=f_2\), but the combined early-selection diversification and pathway-weight effects dominate for larger \(f_1\), the limiting expected Simpson's index first increases and then decreases. This non-monotonicity shows that intratumor heterogeneity is not determined by the magnitude of a driver's fitness gain alone. It also depends on the stage at which the driver is acquired and on how changes in fitness alter the effective representation of the two mutation-order pathways.

To illustrate how these counteracting effects interact, we fix \(f_2\), vary \(f_1\), and report both finite-time simulation results and the corresponding limiting expected Simpson's index in Figure~\ref{fig:combined_simpson_indices}. For the parameter values considered, the late-selection concentration effect along the \((2,1)\) pathway initially dominates, causing the total index to increase. For larger values of \(f_1/f_2\), the early-selection diversification effect along the \((1,2)\) pathway, together with the pathway-weight effect, becomes dominant, and the total index decreases.

\begin{figure}[h]
    \centering
    \subfloat[
        Finite-time distributions of \(R_t\) and \(R_{(12),t}\).
        \label{fig:additive_random_order_total_vs_R12_boxplot_50}
    ]{
        \includegraphics[width=0.48\textwidth]
        {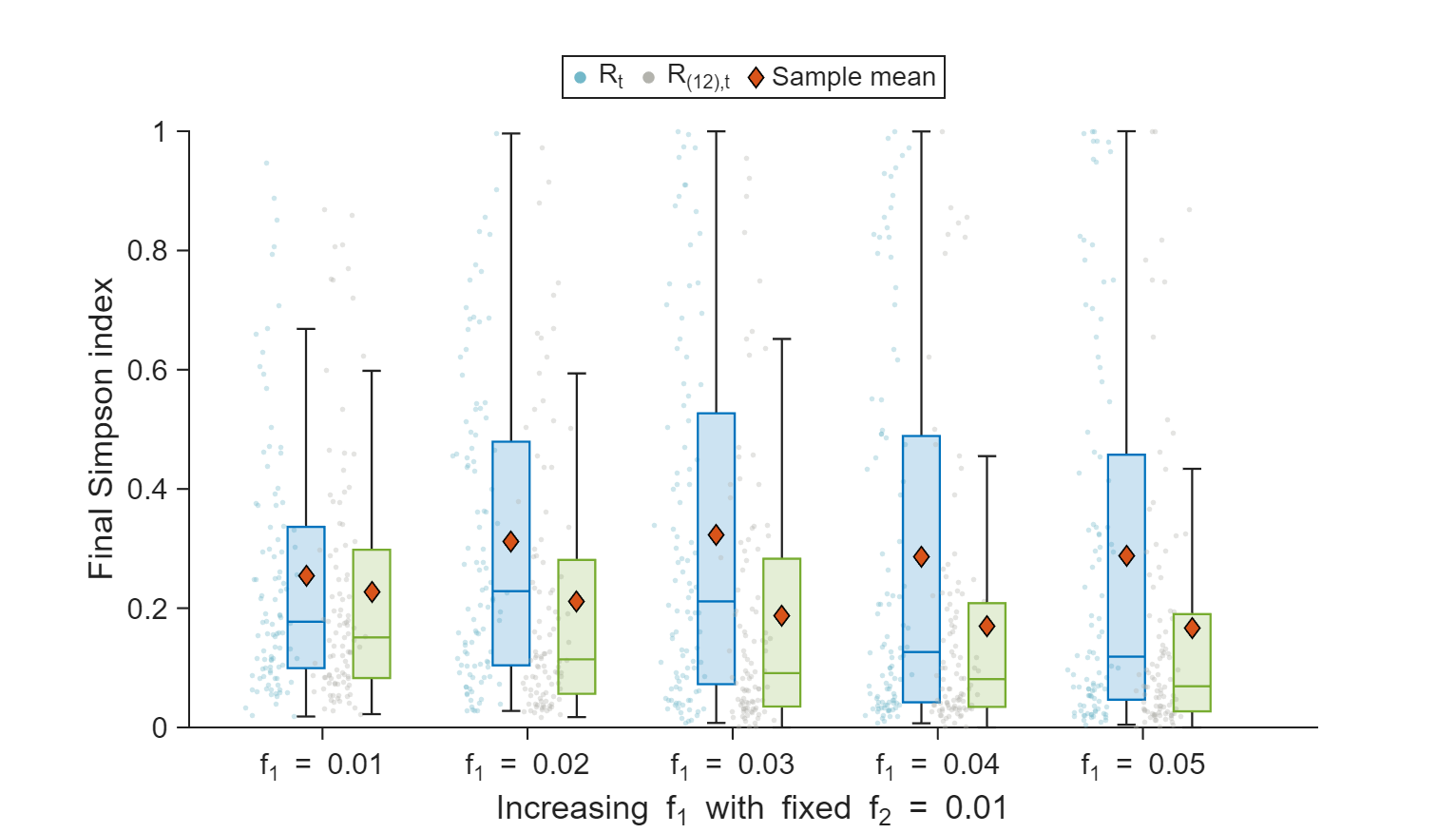}
    }
    \hfill
    \subfloat[
        Limiting expected Simpson's index given by \eqref{eq:additive_simpson_special}.
        \label{fig:Expected_Simpson_index}
    ]{
        \includegraphics[width=0.48\textwidth]
        {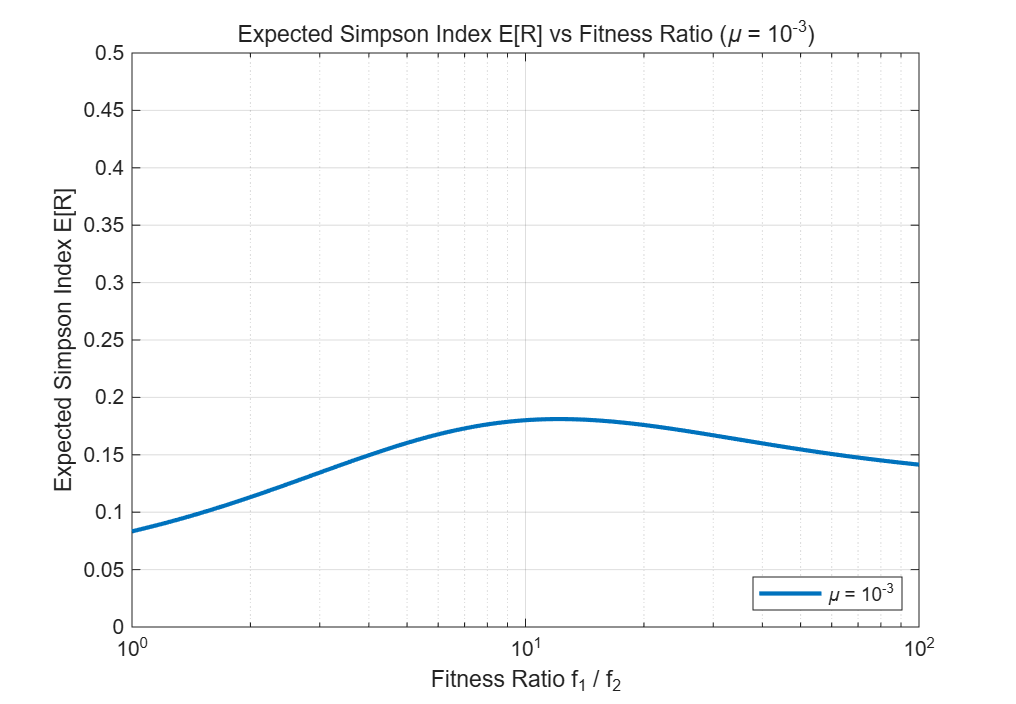}
    }
    \caption{
        Finite-time and limiting Simpson's index under additive fitness effects in the random-order model. Panel~(a) shows the distributions of the Simpson's index \(R_t\) of the total malignant population and the within-pathway Simpson's index \(R_{(12),t}\) at the simulation endpoint, with \(f_2=0.01\) and \(f_1\in\{0.01,0.02,0.03,0.04,0.05\}\). For each value of \(f_1\), the left boxplot corresponds to \(R_t\), whereas the right boxplot corresponds to \(R_{(12),t}\) for malignant clones arising through the \((1,2)\) pathway. Each boxplot summarizes 100 independent simulation realizations. Diamonds denote sample means, and the remaining boxplot conventions and symbols are the same as those in Figure~\ref{fig:deterministic_order_clone_simpson}. As \(f_1\) increases, the within-pathway index decreases, whereas the total-population index first increases and then decreases. Panel~(b) plots the limiting expected Simpson's index given by Equation~\eqref{eq:additive_simpson_special} as a function of the fitness ratio \(f_1/f_2\) over \([1, 10^2]\), with all admissible mutation rates set to \(\mu=10^{-3}\). 
    }
    \label{fig:combined_simpson_indices}
\end{figure}

\textbf{Special case II: commutative but non-additive fitness effects.}

We next consider the random-order model with \emph{commutative but non-additive fitness effects}. In this setting, the selective advantage conferred by a mutation depends on whether the other mutation has already been acquired, even though the two mutation orders ultimately produce malignant cells with the same net growth rate. Such background dependence represents an epistatic
interaction between the two mutations. A canonical biological motivation is the threshold-like inactivation of a tumor suppressor gene. In the classical two-hit model of retinoblastoma, an alteration of one copy of \textit{RB1} creates a predisposed state, whereas tumor initiation typically requires inactivation of the remaining functional copy (\citealt{knudson1971mutation}). Consistent with this mechanism, retinoblastomas generated from human retinal organoids carrying a germline \textit{RB1} alteration were found to have acquired inactivation of the second allele (\citealt{norrie2021retinoblastoma}). 

To represent this threshold-like interaction, suppose that the two single-mutant states have the same net growth rate and that the two double-mutant states also have the same net growth rate. Let \(g_1\) and \(g_2\) denote the fitness advantages of single-mutant and double-mutant cells, respectively, relative to the wild-type population. Specifically, assume $\lambda_{s_1}=\lambda_{s_2}=\lambda_{s_{\emptyset}}+g_1$ and $\lambda_{s_{12}}=\lambda_{s_{21}}=\lambda_{s_{\emptyset}}+g_2$, where \(0<g_1<g_2\). The incremental fitness gain associated with acquiring the second mutation is therefore \(g_2-g_1\). Under symmetric additive fitness effects, the two single mutations would each confer the increment \(g_1\), so the double-mutant advantage would satisfy \(g_2=2g_1\). Allowing \(g_2\neq 2g_1\) makes the interaction non-additive, and the regime $g_2-g_1\gg g_1$ represents a strong positive epistatic effect in which the second alteration has a much larger consequence after the first alteration is already present.

Because
\(\lambda_{s_1}=\lambda_{s_2}=\lambda_{s_{\emptyset}}+g_1\) and
\(\lambda_{s_{12}}=\lambda_{s_{21}}=\lambda_{s_{\emptyset}}+g_2\),
the two pathways have the same limiting within-pathway Simpson's index.
Theorem~\ref{thm:expected_simpson_additive} therefore gives
\begin{align}
\lim_{t\to\infty}\mathbb E[R_t]
&=
\frac{
A_{12}\left(
1-\frac{\lambda_{s_{\emptyset}}+g_1}
        {\lambda_{s_{\emptyset}}+g_2}
\right)
+
A_{21}\left(
1-\frac{\lambda_{s_{\emptyset}}+g_1}
        {\lambda_{s_{\emptyset}}+g_2}
\right)
}{
A_{12}+A_{21}
}
\nonumber\\
&=
1-\frac{\lambda_{s_{\emptyset}}+g_1}
        {\lambda_{s_{\emptyset}}+g_2}
=
\frac{g_2-g_1}{\lambda_{s_{\emptyset}}+g_2}.
\label{eq:nonadditive_symmetric_simpson}
\end{align}
When the fitness effect of a single alteration is small relative to that of the combined alterations, \(g_1\ll g_2\), equation \eqref{eq:nonadditive_symmetric_simpson} becomes
\[
    \lim_{t\to\infty}\mathbb E[R_t]
    \approx
    \frac{g_2}{\lambda_{s_{\emptyset}}+g_2}.
\]
The limiting index is therefore close to one when the double-mutant growth rate is much larger than the single-mutant growth rate. Biologically, a weak first-stage fitness effect limits the expansion of the single-mutant populations and hence the rate at which independent double-mutant founders are produced. Once a
double-mutant clone is founded, however, its substantially larger growth rate gives early-founded clones a strong and persistent advantage over clones founded later. The malignant population can consequently become dominated by one or a few independently founded double-mutant lineages, producing a large Simpson's index. Thus, threshold-like tumor-suppressor inactivation can generate strong lineage-level clonal concentration even when the two possible orders of the inactivating events are selectively symmetric.

Figure~\ref{fig:random_order_nonadditive_clone_simpson_g1_sweep_50} provides finite-time simulation support for the prediction in equation~\eqref{eq:nonadditive_symmetric_simpson}. As \(g_1\) decreases from \(0.05\) to \(0.01\) while \(g_2\) remains fixed at \(0.10\), the sample mean of the total-population Simpson's index increases from \(0.438\) to \(0.765\), and the within-pathway index exhibits the same monotone pattern.

\begin{figure}[h]
    \centering
    \includegraphics[width=0.75\textwidth]{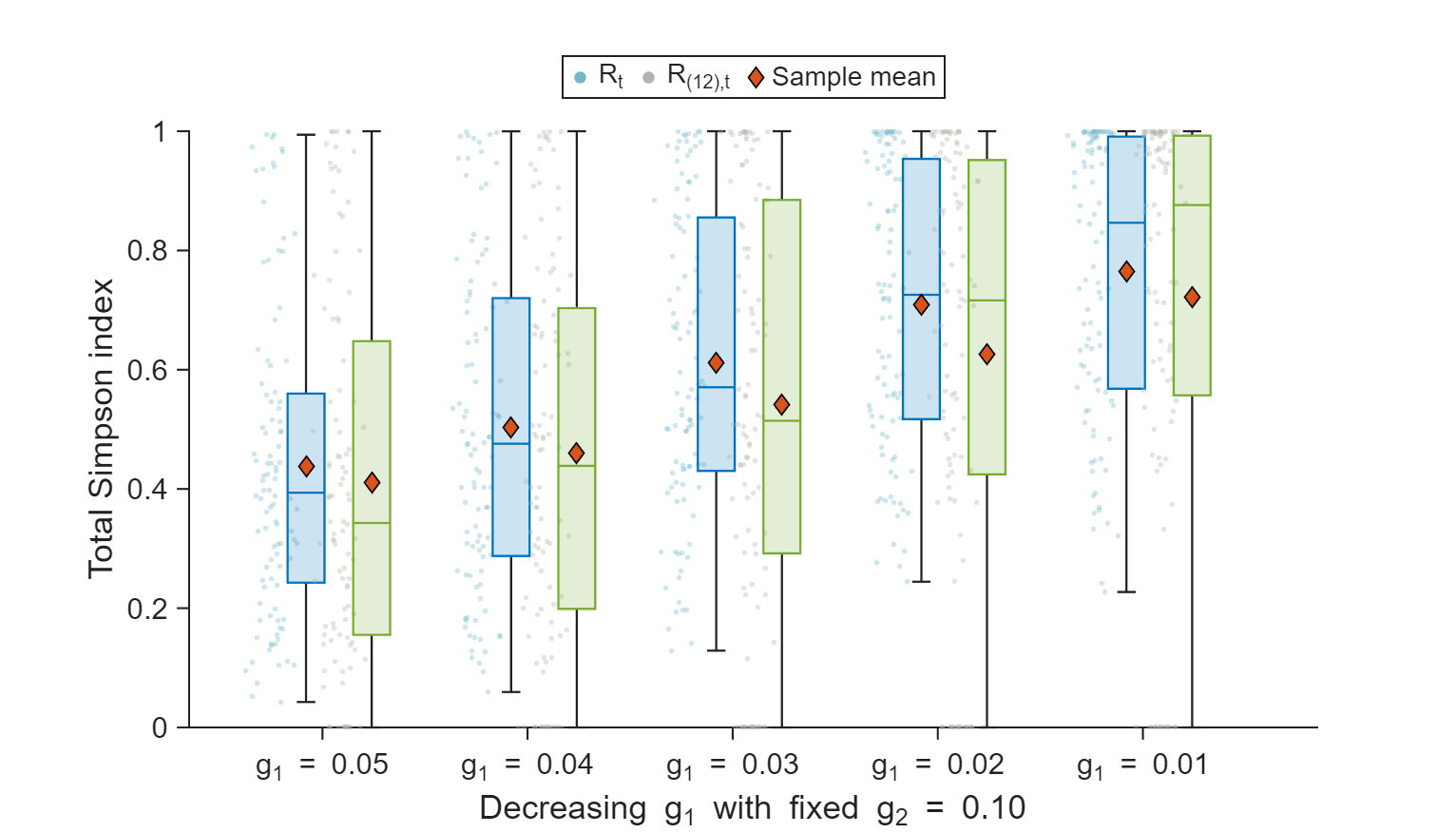}
    \caption{
    Finite-time distributions of the Simpson's index under commutative but non-additive fitness effects. The double-mutant fitness advantage is fixed at \(g_2=0.10\), whereas the single-mutant fitness advantage is varied over \(g_1\in\{0.05,0.04,0.03,0.02,0.01\}\). For each value of \(g_1\), the left boxplot shows the Simpson's index \(R_t\) of the total malignant population, and the right boxplot shows the within-pathway index \(R_{(12),t}\) for malignant clones arising through the \((1,2)\) pathway. Both indices are evaluated at the simulation endpoint. Each boxplot summarizes 100 independent simulation realizations. Diamonds denote sample means, jittered points show individual simulation realizations, and the remaining boxplot conventions are the same as those in Figure~\ref{fig:deterministic_order_clone_simpson}. As \(g_1\) decreases with \(g_2\) fixed, both indices shift toward larger values, indicating greater clonal concentration as the second-stage fitness gain \(g_2-g_1\) increases.}
\label{fig:random_order_nonadditive_clone_simpson_g1_sweep_50}
\end{figure}

\textbf{Extension: non-commutative fitness effects.}

We finally extend the model to allow \emph{non-commutative fitness effects}. In the commutative-fitness regime considered above, the net growth rate of double-mutant cells depends only on the two driver mutations and not on their order of acquisition. Under non-commutative fitness effects, the two mutation-order pathways instead produce double-mutant populations with different net growth rates, so that
\begin{align}\label{eqn:non-commutative fitness}
    \lambda_{s_{12}}\neq\lambda_{s_{21}}.
\end{align}
Condition \eqref{eqn:non-commutative fitness} requires mutational history to remain encoded in a persistent cellular state that is not represented by the two driver mutations alone. Biologically, such history dependence may arise when the first mutation induces or stabilizes an epigenetic or transcriptional program, alters the differentiation state, or produces another persistent cellular change that modifies the effect of the mutation acquired subsequently. In myeloproliferative neoplasms, for example, prior acquisition of a \textit{TET2} mutation was shown to alter the cell-intrinsic transcriptional and proliferative consequences of subsequent \textit{JAK2} V617F acquisition, providing biological support for persistent mutation-order effects (\citealt{ortmann2015effect}).

To isolate the consequence of order-dependent terminal fitness, suppose that $\lambda_{s_{12}}<\lambda_{s_{21}}$, while retaining the assumptions
\(\lambda_{s_1}<\lambda_{s_{12}}\) and \(\lambda_{s_2}<\lambda_{s_{21}}\). This inequality represents a setting in  which the cellular state created by acquiring mutation \(2\) first allows subsequent acquisition of mutation \(1\) to produce a larger terminal growth advantage than that generated through the reverse order. Because the two malignant types  now grow on different exponential scales, Theorem~\ref{thm:expected_simpson_additive}, which relies on equal terminal growth rates, no longer applies.

Specifically, the successive exponential approximation implies $\frac{Z_{12}(t)}{Z_{12}(t)+Z_{21}(t)}\rightarrow 0$. Thus, in the decomposition $R_t
=R_{(12),t}\pi_{12}(t)^2+R_{(21),t}\pi_{21}(t)^2$, the contribution of the lower-growth \((1,2)\) pathway vanishes. Applying Proposition~\ref{lem:deterministic_simpson} to the transition \(s_2\to s_{21}\) therefore gives, under this establishment conditioning,
\[
    \lim_{t\to\infty}\mathbb E[R_t]
    =
    1-\frac{\lambda_{s_2}}{\lambda_{s_{21}}}.
\]
Mutation order consequently affects heterogeneity through a qualitatively different mechanism from that operating in the commutative-fitness regime. When the terminal growth rates are equal, both mutation orders can remain represented on the same asymptotic scale, and the limiting Simpson's index balances pathway representation against clonal concentration within each pathway. When the terminal growth rates differ, by contrast, the pathway with the larger rate eventually accounts for essentially all malignant cells, and the limiting index is determined by the clone-founding and growth dynamics along that pathway alone.

\section{Concluding Comments}\label{Sec:conclusion}

In this paper, we developed a multitype branching-process framework to study how pathway-specific growth and mutation dynamics shape the clonal structure of a malignant population when driver mutations can be acquired in different orders. By distinguishing malignant cells according to both their mutation-order pathway and their independent founding event, the model separates heterogeneity across evolutionary pathways from heterogeneity among clones within each pathway. Our results show that clonal concentration depends not only on the selective advantages conferred by driver mutations, but also on when those advantages arise during progression, how rapidly intermediate populations generate malignant founders, and how strongly the resulting pathways are represented in the malignant population. In particular, a stronger driver need not always reduce clonal diversity, because selection acting at an intermediate stage can expand the population from which independent malignant clones arise.

The analysis has several limitations. We consider well-mixed populations with exponential growth, and characterize clonal structure primarily through the expected Simpson's index. The main results are also asymptotic and rely on the successive exponential approximation. These assumptions provide analytical tractability and isolate the mechanisms of interest, but they do not capture the spatial constraints, resource limitations, differentiation structure, or treatment-induced changes present in many tumors. Moreover, finite-time clonal structure may differ from its asymptotic behavior, particularly when competing malignant populations have similar growth rates.

Future work could extend the framework to mutation networks involving additional drivers, shared intermediate states, and multiple malignant genotypes. It would also be useful to establish finite-time results or corresponding limits directly for the original multitype branching process, and to study the full distribution of clonal-diversity measures rather than their expectations alone. Incorporating spatial growth, ecological interactions, and treatment would broaden the biological scope of the model. Finally, lineage-tracing data, single-cell sequencing, and tumor phylogenies could be used to estimate pathway-specific growth and mutation parameters and test the predicted relationship between evolutionary dynamics and clonal structure. More broadly, our analysis suggests that intratumor heterogeneity should be interpreted in relation to the evolutionary routes through which malignant clones are generated, rather than solely through the identities or fitness effects of the mutations they carry.

\newpage

\appendix

\section{Proof of Proposition~\ref{lem:deterministic_simpson}}
\label{app:proof_deterministic_simpson}

\begin{proof}
By the within-generation heterogeneity result of \citealt{durrett2011intratumor} (Equation~(11)), the limiting Simpson's index \(R_{\infty}\) satisfies $\mathbb E[R_{\infty}]=1-\lambda_{s_1}/\lambda_{s_{12}}$. By Lemma \ref{lem:one_step_clone_process} and Lemma \ref{lem:simpson_convergence}, $\lim_{t\to\infty}\mathbb E[R_t]=\mathbb E[R_{\infty}]$.
\end{proof}

\section{Proof of Theorem~\ref{thm:expected_simpson_additive}}
\label{app:thm:expected_simpson_additive}

Before proving Theorem~\ref{thm:expected_simpson_additive}, we establish three auxiliary lemmas. The first characterizes the clone-size point process generated by a single mutational transition from an exponentially growing source population. It also establishes that the random coefficient \(V_{s_j}\) in the asymptotic exponential representation of the \(s_j\) population is given by the total mass of the limiting clone-size point process.

Consider an admissible mutational transition \(s_i\to s_j\) satisfying $\lambda_{s_i}<\lambda_{s_j}$. Suppose that the \(s_i\) population is represented by $Z^*_{s_i}(u)=V_{s_i}e^{\lambda_{s_i}u}$ for $u\in\mathbb R$. Let \(U_k\) denote the time of the \(k\)-th \(s_i\to s_j\) mutation event, and let \(C_k(r)\) denote the size, at clone age \(r\), of the \(s_j\) clone founded by that event. Define the point process of rescaled \(s_j\)-clone sizes at time \(t\) by
\[
\mathcal P_{i,j}^{(t)}
=
\sum_{\substack{k:U_k\leq t\\ C_k(t-U_k)>0}}
\delta_{e^{-\lambda_{s_j}t}C_k(t-U_k)},
\]
where extinct clones are omitted and \(\delta_x\) denotes the unit point mass at \(x\). Let \(\mathcal M_p((0,\infty))\) denote the space of locally finite point measures on \((0,\infty)\), equipped with the vague topology. The notation \(\Rightarrow\) denotes convergence in distribution in this space. We then have the following result.

\begin{lemma}\label{lem:one_step_clone_process}
Conditional on \(V_{s_i}\), $\mathcal P_{i,j}^{(t)}\Rightarrow\mathcal P_{i,j,V_{s_i}}$ in the vague topology on \(\mathcal M_p((0,\infty))\), where \(\mathcal P_{i,j,V_{s_i}}\) is a Poisson point process with intensity measure 
\begin{equation}\label{eq:one_step_levy_measure}
    \nu_{i,j,V_{s_i}}(dx)=\frac{\mu_{i,j}V_{s_i} H_{i,j}\beta}{\Gamma(1-\beta)}x^{-\beta-1}\,dx,
\end{equation}
where $\beta=\frac{\lambda_{s_i}}{\lambda_{s_j}}$. The total mass of the limiting point process, $V_{s_j}:=\int_{(0,\infty)}x\,\mathcal P_{i,j,V_{s_i}}(dx)$,
is finite almost surely. Its conditional Laplace exponent is
\[
-\log
\mathbb E\!\left[
e^{-\theta V_{s_j}}
\,\middle|\,
V_{s_i}
\right]
=
\mu_{i,j}V_{s_i}H_{i,j}\theta^\beta,
\qquad \theta\geq0.
\]
\end{lemma}

\subsection{Proof of Lemma \ref{lem:one_step_clone_process}}

\begin{proof}
Because the type $s_i$ population is $Z^*_{s_i}(u)=V_{s_i}e^{\lambda_{s_i}u}$, the times \(\{U_k\}_{k\geq 1}\) of the \(s_i\to s_j\) mutation events form an inhomogeneous Poisson point process on \(\mathbb R\) with intensity $\mu_{i,j}V_{s_i} e^{\lambda_{s_i}u}\,du$. For each \(k\), let
\[
\xi_k
=
\lim_{r\to\infty}
e^{-\lambda_{s_j}r}C_k(r)
\]
denote the martingale limit of the \(s_j\) clone founded at time \(U_k\). The random variables \(\{\xi_k\}_{k\geq 1}\) are independent copies of \(\xi_{s_j}\), where $\mathbb P(\xi_{s_j}=0)=\frac{b_{s_j}}{a_{s_j}}$, and, conditional on \(\xi_{s_j}>0\), $\xi_{s_j}\sim\operatorname{Exp}\left(\frac{\lambda_{s_j}}{a_{s_j}}\right)$. For each fixed mutation time \(U_k\), $e^{-\lambda_{s_j}t}C_k(t-U_k)$ converges to $e^{-\lambda_{s_j}U_k}\xi_k$ a.s. as $t\to\infty$. \sout{Thus, the limiting point process can be represented as} Therefore, we define a limiting point process as follows
\[
\mathcal P_{i,j,V_{s_i}}
:=
\sum_{k:\xi_k>0}
\delta_{e^{-\lambda_{s_j}U_k}\xi_k}.
\]
By the Poisson marking theorem, $\sum_k\delta_{(U_k,\xi_k)}$ is a Poisson point process on \(\mathbb R\times(0,\infty)\) with intensity measure $\mu_{i,j}V_{s_i} e^{\lambda_{s_i}u}\,du\,\mathbb P(\xi_{s_j}\in dz)$\footnote{Abusing notation, $\xi_{s_j}\sim\operatorname{Exp}\left(\frac{\lambda_{s_j}}{a_{s_j}}\right)$.}. Applying the Poisson mapping theorem to the map $(u,z)\rightarrow e^{-\lambda_{s_j}u}z$ shows that \(\mathcal P_{i,j,V_{s_i}}\) is a Poisson point process on \((0,\infty)\). 

For \(x>0\), we have
\begin{align*}
\nu_{i,j,V_{s_i}}\bigl((x,\infty)\bigr)
&=
\mu_{i,j}V_{s_i}
\int_{-\infty}^{\infty}
e^{\lambda_{s_i}u}
\mathbb P\left(
e^{-\lambda_{s_j}u}\xi_{s_j}>x
\right)du
\\
&=
\mu_{i,j}V_{s_i}
\int_{-\infty}^{\infty}
e^{\lambda_{s_i}u}
\mathbb P\left(
\xi_{s_j}>xe^{\lambda_{s_j}u}
\right)du
\\
&=
\mu_{i,j}V_{s_i}
\frac{\lambda_{s_j}}{a_{s_j}}
\int_{-\infty}^{\infty}
e^{\lambda_{s_i}u}
\exp\left\{
-\frac{\lambda_{s_j}}{a_{s_j}}
xe^{\lambda_{s_j}u}
\right\}du.
\end{align*}
Using the change of variables $z=\frac{\lambda_{s_j}}{a_{s_j}}xe^{\lambda_{s_j}u}$, we obtain
\begin{align*}
\nu_{i,j,V_{s_i}}\bigl((x,\infty)\bigr)
&=
\frac{\mu_{i,j}V_{s_i}}{a_{s_j}}
\left(
\frac{a_{s_j}}{\lambda_{s_j}}
\right)^\beta
x^{-\beta}
\int_0^\infty
z^{\beta-1}e^{-z}\,dz
\\
&=
\frac{\mu_{i,j}V_{s_i}}{a_{s_j}}
\left(
\frac{a_{s_j}}{\lambda_{s_j}}
\right)^\beta
\Gamma(\beta)x^{-\beta}.
\end{align*}
By the definition of \(H_{i,j}\) in
\eqref{eq:H_constant} and Euler's reflection identity, $\Gamma(\beta)\Gamma(1-\beta)=\frac{\pi}{\sin(\pi\beta)}$, we have $\frac{1}{a_{s_j}}\left(\frac{a_{s_j}}{\lambda_{s_j}}\right)^\beta\Gamma(\beta)=\frac{H_{i,j}}{\Gamma(1-\beta)}$. Therefore, $\nu_{i,j,V_{s_i}}\bigl((x,\infty)\bigr)=\frac{\mu_{i,j}V_{s_i} H_{i,j}}{\Gamma(1-\beta)}x^{-\beta}$. Differentiating the tail measure gives
\[
\nu_{i,j,V_{s_i}}(dx)
=
\frac{
\mu_{i,j}V_{s_i} H_{i,j}\beta
}{
\Gamma(1-\beta)
}
x^{-\beta-1}\,dx,
\]
which proves \eqref{eq:one_step_levy_measure}.

We next prove the vague convergence of
\(\mathcal P_{i,j}^{(t)}\). Let \(C_c((0,\infty))\) denote the space of real-valued continuous functions on \((0,\infty)\) with compact support, and let \(f\in C_c((0,\infty))\). Since the support of \(f\) is bounded away from zero, there exists \(\varepsilon>0\) such that $f(x)=0$ for $0<x<\varepsilon$. We extend \(f\) to zero by setting \(f(0)=0\). Fix \(K\in\mathbb R\). The number of mutation events satisfying \(U_k\leq K\) is Poisson distributed with finite mean $\mu_{i,j}V_{s_i}\int_{-\infty}^{K}e^{\lambda_{s_i}u}\,du=\frac{\mu_{i,j}V_{s_i}}{\lambda_{s_i}}e^{\lambda_{s_i}K}$. Hence, almost surely, only finitely many mutation events occur before time \(K\). For these finitely many clones, the clone-wise martingale convergence implies $\sum_{k:U_k\leq K}f\left(e^{-\lambda_{s_j}t}C_k(t-U_k)
\right)\rightarrow\sum_{k:U_k\leq K}f\left(e^{-\lambda_{s_j}U_k}\xi_k\right)$ almost surely as \(t\to\infty\).

It remains to control clones founded after time \(K\). By the Poisson mutation intensity,
\begin{align*}
\mathbb E\left[
\sum_{k:K<U_k\leq t}
\mathbf 1
\left\{
e^{-\lambda_{s_j}t}
C_k(t-U_k)
\geq\varepsilon
\right\}
\,\middle|\,
V_{s_i}
\right]
=
\mu_{i,j}V_{s_i}
\int_K^t
e^{\lambda_{s_i}u}
\mathbb P\left(
e^{-\lambda_{s_j}t}
C_1(t-u)
\geq\varepsilon
\right)du.
\end{align*}
Since $\mathbb E[C_1(r)]=e^{\lambda_{s_j}r}$, Markov's inequality yields $\mathbb P\left(e^{-\lambda_{s_j}t}C_1(t-u)\geq\varepsilon\right)\leq\frac{e^{-\lambda_{s_j}u}}{\varepsilon}$. It follows that
\begin{align*}
\mathbb E\left[
\sum_{k:K<U_k\leq t}
\mathbf 1
\left\{
e^{-\lambda_{s_j}t}
C_k(t-U_k)
\geq\varepsilon
\right\}
\,\middle|\,
V_{s_i}
\right]
&\leq
\frac{\mu_{i,j}V_{s_i}}{\varepsilon}
\int_K^t
e^{-(\lambda_{s_j}-\lambda_{s_i})u}\,du
\\
&\leq
\frac{\mu_{i,j}V_{s_i}}
{\varepsilon(\lambda_{s_j}-\lambda_{s_i})}
e^{-(\lambda_{s_j}-\lambda_{s_i})K}.
\end{align*}
Similarly, using \(\mathbb E[\xi_{s_j}]=1\),
\begin{align*}
\mathbb E\left[
\sum_{k:U_k>K}
\mathbf 1
\left\{
e^{-\lambda_{s_j}U_k}\xi_k
\geq\varepsilon
\right\}
\,\middle|\,
V_{s_i}
\right]
&=
\mu_{i,j}V_{s_i}
\int_K^\infty
e^{\lambda_{s_i}u}
\mathbb P\left(
e^{-\lambda_{s_j}u}
\xi_{s_j}
\geq\varepsilon
\right)du
\\
&\leq
\frac{\mu_{i,j}V_{s_i}}
{\varepsilon(\lambda_{s_j}-\lambda_{s_i})}
e^{-(\lambda_{s_j}-\lambda_{s_i})K}.
\end{align*}
Both bounds converge to zero as \(K\to\infty\), uniformly in \(t\geq K\). Combining these concentration results with the almost-sure convergence of the finitely many clones founded before \(K\), we obtain
\[
\int_{(0,\infty)}
f(x)\,\mathcal P_{i,j}^{(t)}(dx)
\rightarrow
\int_{(0,\infty)}
f(x)\,\mathcal P_{i,j,V_{s_i}}(dx)
\]
in probability conditional on \(V_{s_i}\). Applying this to a countable convergence-determining family in \(f\in C_c((0,\infty))\) shows that $\mathcal P_{i,j}^{(t)} \rightarrow P_{i,j,V_{s_i}}$ in probability in the vague topology, conditional on \(V_{s_i}\), and hence also in distribution.

Finally, let $V_{s_j}=\int_{(0,\infty)}x\,\mathcal P_{i,j,V_{s_i}}(dx)$. Since \(0<\beta<1\), $\int_0^\infty(1\wedge x)\nu_{i,j,V_{s_i}}(dx)<\infty$, and hence \(V_{s_j}<\infty\) almost surely. By the Laplace functional of a Poisson point process,
\[
\mathbb E\left[
e^{-\theta V_{s_j}}
\,\middle|\,
V_{s_i}
\right]
=
\exp\left\{
-
\int_0^\infty
\left(1-e^{-\theta x}\right)
\nu_{i,j,V_{s_i}}(dx)
\right\}.
\]
Therefore,
\begin{align*}
-\log
\mathbb E\left[
e^{-\theta V_{s_j}}
\,\middle|\,
V_{s_i}
\right]
&=
\frac{
\mu_{i,j}V_{s_i} H_{i,j}\beta
}{
\Gamma(1-\beta)
}
\int_0^\infty
\left(1-e^{-\theta x}\right)
x^{-\beta-1}\,dx.
\end{align*} 
For \(0<\beta<1\), $\int_0^\infty\left(1-e^{-\theta x}\right)x^{-\beta-1}\,dx=\frac{\Gamma(1-\beta)}{\beta}\theta^\beta$. Consequently,
\[
-\log
\mathbb E\left[
e^{-\theta V_{s_j}}
\,\middle|\,
V_{s_i}
\right]
=
\mu_{i,j}V_{s_i} H_{i,j}\theta^\beta,
\]
as desired.
\end{proof}

To obtain the Simpson's index of the total malignant population, we combine the limiting clone-size point processes generated along the \((1,2)\) and \((2,1)\) pathways. Because the two malignant cell types have the same net growth rate \(\lambda\), both pathway-specific populations contribute on the same asymptotic scale. Accordingly, the finite-time decomposition $R_t=R_{(12),t}\pi_{12}(t)^2+R_{(21),t}\pi_{21}(t)^2$ carries over to the point-process limit. The limiting Simpson's index therefore depends jointly on the clonal diversity within each pathway and on the random fraction of the malignant population contributed by that pathway. These quantities are generally not independent, so their expectations cannot be evaluated separately. The following lemma performs the required joint calculation and expresses the limiting expected Simpson's index in terms of the pathway-scale constants \(A_{12}\) and \(A_{21}\).

Specifically, let \(\mathcal P_{1,12,V_{s_1}}\) and \(\mathcal P_{2,21,V_{s_2}}\) denote the limiting clone-size point processes corresponding to the two mutation orders \(s_{\emptyset}\to s_1\to s_{12}\) and \(s_{\emptyset}\to s_2\to s_{21}\). For simplicity of notation, we use \(\mathcal P_{12}\) and \(\mathcal P_{21}\) instead. The points of \(\mathcal P_{12}\) represent the limiting scaled sizes of individual \(s_{12}\)-clones, and the points of \(\mathcal P_{21}\) represent the limiting scaled sizes of individual \(s_{21}\)-clones. Let $S_{12}=\sum_{x\in\mathcal P_{12}}x$, $S_{21}=\sum_{y\in\mathcal P_{21}}y$, $S_{12}^{(2)}=\sum_{x\in\mathcal P_{12}}x^2$, and $S_{21}^{(2)}=\sum_{y\in\mathcal P_{21}}y^2$. We define the limiting Simpson's index of the total malignant population as
\[
R=
\frac{S_{12}^{(2)}+S_{21}^{(2)}}{(S_{12}+S_{21})^2}.
\]
The total mass \(S_{12}+S_{21}\) is positive almost surely on the non-extinction event considered under the successive exponential approximation.

\begin{lemma}
\label{lem:limiting_expected_simpson}
Let \(R\) be the limiting Simpson's index of the total malignant population
under the point-process limit. Then
\begin{equation}
\label{eq:limiting_expected_R}
\mathbb E[R]
=
\frac{
A_{12}\left(1-\frac{\lambda_{s_1}}{\lambda}\right)
+
A_{21}\left(1-\frac{\lambda_{s_2}}{\lambda}\right)
}{
A_{12}+A_{21}
}.
\end{equation}
\end{lemma}

\subsection{Proof of Lemma \ref{lem:limiting_expected_simpson}}

\begin{proof}

For compactness in the proof only, write \(D_{12}=\mu_{1,12}H_{1,12}\), \(D_{21}=\mu_{2,21}H_{2,21}\),
\(\alpha_{12}=\lambda_{s_1}/\lambda\), and \(\alpha_{21}=\lambda_{s_2}/\lambda\). Since \(\lambda_{s_1},\lambda_{s_2}<\lambda\), we have \(0<\alpha_{12},\alpha_{21}<1\). Conditional on \(V_{s_1}\) and \(V_{s_2}\), the point processes \(\mathcal P_{12}\) and \(\mathcal P_{21}\) are independent Poisson point processes. By Lemma \ref{lem:one_step_clone_process}, their total masses have Laplace exponents $\Psi_{12}(\theta)=D_{12}V_{s_1}\theta^{\alpha_{12}}$ and $\Psi_{21}(\theta)=D_{21}V_{s_2}\theta^{\alpha_{21}}$. 

We first compute the conditional expectation \(\mathbb E[R\mid V_{s_1},V_{s_2}]\). For \(z>0\), 
\begin{equation}\label{eqn:transform_simpson}
    \frac{1}{z^2}=\int_0^\infty r e^{-rz}\,dr.
\end{equation}
Since all terms in the following calculation are nonnegative, Tonelli's theorem allows us to interchange expectation and integration. Thus, by \eqref{eqn:transform_simpson},
\[
\mathbb E[R\mid V_{s_1},V_{s_2}]
=
\int_0^\infty
r\,
\mathbb E\left[
\left(S_{12}^{(2)}+S_{21}^{(2)}\right)
e^{-r(S_{12}+S_{21})}
\mid V_{s_1},V_{s_2}
\right]dr .
\]
We now evaluate the expectation inside the integral. We use the Campbell--Mecke formula for Poisson point processes (\citealt{last2017lectures}) (Theorem~4.1). If \(\mathcal P\) is a Poisson point process on \((0,\infty)\) with mean measure \(\nu\), then for any nonnegative measurable function \(F(x,\mathcal P)\),
\[
\mathbb E\left[
\sum_{x\in\mathcal P}F(x,\mathcal P)
\right]
=
\int_0^\infty
\mathbb E\left[
F(x,\mathcal P+\delta_x)
\right]\nu(dx).
\]
Let \(\mathcal P\) be a Poisson point process with total mass \(S=\sum_{x\in\mathcal P}x\), mean measure \(\nu\), and Laplace exponent $\Psi(r)=\int_0^\infty(1-e^{-rx})\nu(dx)$. Applying the Campbell--Mecke formula with \(F(x,\mathcal P)=x^2\exp\{-r\sum_{z\in\mathcal P}z\}\), we obtain $\mathbb E\left[\sum_{x\in\mathcal P}x^2e^{-rS}\right]=\int_0^\infty x^2e^{-rx}\mathbb E[e^{-rS}]\nu(dx)$. Since the Laplace transform of \(S\) is \(\mathbb E[e^{-rS}]=e^{-\Psi(r)}\), this becomes $\mathbb E\left[\sum_{x\in\mathcal P}x^2e^{-rS}\right]=e^{-\Psi(r)}\int_0^\infty x^2e^{-rx}\nu(dx)$. Moreover, we have $\Psi'(r)=\int_0^\infty xe^{-rx}\nu(dx)$ and $\Psi''(r)=-\int_0^\infty x^2e^{-rx}\nu(dx)$. Therefore
\begin{equation}\label{eqn:identity}
    \mathbb E\left[
\sum_{x\in\mathcal P}x^2e^{-rS}
\right]
=
e^{-\Psi(r)}[-\Psi''(r)].
\end{equation}

We then apply \eqref{eqn:identity} separately to \(\mathcal P_{12}\) and \(\mathcal P_{21}\). Conditional on \(V_{s_1}\) and \(V_{s_2}\), the two point processes are independent. Hence
\begin{align*}
    \mathbb E\left[
S_{12}^{(2)}e^{-r(S_{12}+S_{21})}
\mid V_{s_1},V_{s_2}
\right] & =\mathbb E\left[
S_{12}^{(2)}e^{-rS_{12}}
\mid V_{s_1}
\right]
\mathbb E\left[
e^{-rS_{21}}
\mid V_{s_2}
\right]\\
& =e^{-\Psi_{12}(r)}[-\Psi_{12}''(r)]e^{-\Psi_{21}(r)}.
\end{align*}
Similarly,
\[
\mathbb E\left[
S_{21}^{(2)}e^{-r(S_{12}+S_{21})}
\mid V_{s_1},V_{s_2}
\right]
=
e^{-\Psi_{12}(r)}e^{-\Psi_{21}(r)}[-\Psi_{21}''(r)].
\]
Adding the two terms gives
\[
\begin{aligned}
&\mathbb E\left[
\left(S_{12}^{(2)}+S_{21}^{(2)}\right)
e^{-r(S_{12}+S_{21})}
\mid V_{s_1},V_{s_2}
\right]  =
e^{-\Psi_{12}(r)-\Psi_{21}(r)}
\left[-\Psi_{12}''(r)-\Psi_{21}''(r)\right].
\end{aligned}
\]
Since \(\Psi_{12}(r)=D_{12}V_{s_1}r^{\alpha_{12}}\), we have $-\Psi_{12}''(r)=D_{12}V_{s_1}\alpha_{12}(1-\alpha_{12})r^{\alpha_{12}-2}$. Similarly, $-\Psi_{21}''(r)=D_{21}V_{s_2}\alpha_{21}(1-\alpha_{21})r^{\alpha_{21}-2}$. Multiplying by the extra factor \(r\) from
\(z^{-2}=\int_0^\infty r e^{-rz}dr\), we get
\begin{equation}
\label{eq:conditional_expected_R}
\begin{aligned}
\mathbb E[R\mid V_{s_1},V_{s_2}]
=
\int_0^\infty
&
e^{-D_{12}V_{s_1}r^{\alpha_{12}}
  -D_{21}V_{s_2}r^{\alpha_{21}}}
\\
&\times
\left[
D_{12}V_{s_1}\alpha_{12}(1-\alpha_{12})r^{\alpha_{12}-1}
+
D_{21}V_{s_2}\alpha_{21}(1-\alpha_{21})r^{\alpha_{21}-1}
\right]dr .
\end{aligned}
\end{equation}

We now average over \(V_{s_1}\) and \(V_{s_2}\), conditional on \(V_0\). By Lemma~\ref{lem:one_step_clone_process} applied to \(s_{\emptyset}\to s_1\) and \(s_{\emptyset}\to s_2\), the first-stage intermediate limits satisfy $\mathbb E[e^{-\eta V_{s_1}}\mid V_0]=\exp\left\{-\mu_{\emptyset,1}H_{\emptyset,1}V_0\eta^{\lambda_{s_{\emptyset}}/\lambda_{s_1}}\right\}$ and $\mathbb E[e^{-\eta V_{s_2}}\mid V_0]=\exp\left\{-\mu_{\emptyset,2}H_{\emptyset,2}V_0\eta^{\lambda_{s_{\emptyset}}/\lambda_{s_2}}\right\}$.
The two random variables \(V_{s_1}\) and \(V_{s_2}\) are independent conditional on \(V_0\), because they are generated by independent mutation processes from the deterministic wild-type population.

We compute the contribution from the term $e^{-D_{12}V_{s_1}r^{\alpha_{12}}-D_{21}V_{s_2}r^{\alpha_{21}}} \times D_{12}V_{s_1}\alpha_{12}(1-\alpha_{12})r^{\alpha_{12}-1}$ in detail. This term can be decomposed into the product of 
\begin{align}
& e^{-D_{12}V_{s_1}r^{\alpha_{12}}} \times D_{12}V_{s_1}\alpha_{12}(1-\alpha_{12})r^{\alpha_{12}-1}, \text{ and} \label{eqn: first_term}\\
& e^{-D_{21}V_{s_2}r^{\alpha_{21}}}. \label{eqn: second_term}
\end{align}
Differentiating the conditional Laplace transform of \(V_{s_1}\) gives
\[
\mathbb E\left[V_{s_1}e^{-\eta V_{s_1}}\mid V_0\right]
=
\mu_{\emptyset,1}H_{\emptyset,1}V_0
\frac{\lambda_{s_{\emptyset}}}{\lambda_{s_1}}
\eta^{\lambda_{s_{\emptyset}}/\lambda_{s_1}-1}
\exp\left\{
-\mu_{\emptyset,1}H_{\emptyset,1}V_0
\eta^{\lambda_{s_{\emptyset}}/\lambda_{s_1}}
\right\}.
\]
Recalling that \(A_{12}=\mu_{\emptyset,1}H_{\emptyset,1} D_{12}^{\lambda_{s_{\emptyset}}/\lambda_{s_1}}\), then we have
\[
\mathbb E\left[
D_{12}V_{s_1}e^{-D_{12}V_{s_1}r^{\alpha_{12}}}
\mid V_0
\right]
=
A_{12}V_0
\frac{\lambda_{s_{\emptyset}}}{\lambda_{s_1}}
r^{\lambda_{s_{\emptyset}}/\lambda-\lambda_{s_1}/\lambda}
e^{-A_{12}V_0r^{\lambda_{s_{\emptyset}}/\lambda}} .
\]
Multiplying by the remaining factor \(\alpha_{12}(1-\alpha_{12})r^{\alpha_{12}-1}\) from \eqref{eq:conditional_expected_R}, and using \((\lambda_{s_{\emptyset}}/\lambda_{s_1})\alpha_{12} =\lambda_{s_{\emptyset}}/\lambda\), the \eqref{eqn: first_term} contribution becomes $A_{12}V_0
\frac{\lambda_{s_{\emptyset}}}{\lambda}\left(1-\frac{\lambda_{s_1}}{\lambda}\right)r^{\lambda_{s_{\emptyset}}/\lambda-1}e^{-A_{12}V_0r^{\lambda_{s_{\emptyset}}/\lambda}}$. Additionally, \eqref{eqn: second_term} contributes the factor $\mathbb E\left[e^{-D_{21}V_{s_2}r^{\alpha_{21}}}\mid V_0\right]=e^{-A_{21}V_0r^{\lambda_{s_{\emptyset}}/\lambda}}$. Thus the total contribution is
\[
A_{12}V_0
\frac{\lambda_{s_{\emptyset}}}{\lambda}
\left(1-\frac{\lambda_{s_1}}{\lambda}\right)
r^{\lambda_{s_{\emptyset}}/\lambda-1}
e^{-(A_{12}+A_{21})V_0r^{\lambda_{s_{\emptyset}}/\lambda}} .
\]
The contribution from the other term in \eqref{eq:conditional_expected_R} is obtained in the same way:
\[
A_{21}V_0
\frac{\lambda_{s_{\emptyset}}}{\lambda}
\left(1-\frac{\lambda_{s_2}}{\lambda}\right)
r^{\lambda_{s_{\emptyset}}/\lambda-1}
e^{-(A_{12}+A_{21})V_0r^{\lambda_{s_{\emptyset}}/\lambda}} .
\]
Substituting these two contributions into
\eqref{eq:conditional_expected_R} gives
\begin{equation}
\label{eq:conditional_on_V0}
\begin{aligned}
\mathbb E[R\mid V_0]
=
\int_0^\infty
&
V_0\frac{\lambda_{s_{\emptyset}}}{\lambda}
r^{\lambda_{s_{\emptyset}}/\lambda-1}
e^{-(A_{12}+A_{21})V_0r^{\lambda_{s_{\emptyset}}/\lambda}}
\\
&\times
\left[
A_{12}\left(1-\frac{\lambda_{s_1}}{\lambda}\right)
+
A_{21}\left(1-\frac{\lambda_{s_2}}{\lambda}\right)
\right]dr .
\end{aligned}
\end{equation}

It remains only to evaluate the integral in \eqref{eq:conditional_on_V0}. Let \(w=(A_{12}+A_{21})V_0r^{\lambda_{s_{\emptyset}}/\lambda}\). Then $dw=(A_{12}+A_{21})V_0
\frac{\lambda_{s_{\emptyset}}}{\lambda}r^{\lambda_{s_{\emptyset}}/\lambda-1}dr$. Therefore
\[
\int_0^\infty
V_0\frac{\lambda_{s_{\emptyset}}}{\lambda}
r^{\lambda_{s_{\emptyset}}/\lambda-1}
e^{-(A_{12}+A_{21})V_0r^{\lambda_{s_{\emptyset}}/\lambda}}
dr
=
\frac{1}{A_{12}+A_{21}}.
\]
Thus
\[
\mathbb E[R\mid V_0]
=
\frac{
A_{12}\left(1-\frac{\lambda_{s_1}}{\lambda}\right)
+
A_{21}\left(1-\frac{\lambda_{s_2}}{\lambda}\right)
}{
A_{12}+A_{21}
}.
\]
The right-hand side no longer depends on \(V_0\). Taking expectation again gives \eqref{eq:limiting_expected_R}.
\end{proof}

Lemma~\ref{lem:limiting_expected_simpson} identifies the expected Simpson's index associated with the limiting clone-size point process. To apply this result to the finite-time population under the successive exponential approximation, it remains to connect the Simpson's index \(R_t\) of the malignant population at time \(t\) to its limiting point-process counterpart \(R\).

\begin{lemma}\label{lem:simpson_convergence}
Under the assumptions of Theorem~\ref{thm:expected_simpson_additive},
\(\lim_{t\to\infty}\mathbb E[R_t]=\mathbb E[R]\), where \(R\) is the limiting
Simpson's index in Lemma~\ref{lem:limiting_expected_simpson}.
\end{lemma}

\subsection{Proof of Lemma \ref{lem:simpson_convergence}}

\begin{proof}
Recall that we denote by \(X_i(t)\) the number of cells at time \(t\) in the \(i\)-th malignant clone arising along the \((1,2)\) pathway, and by \(Y_j(t)\) the corresponding size of the \(j\)-th malignant clone arising along the \((2,1)\) pathway. Let $x_i^{(t)}=e^{-\lambda t}X_i(t)$ and $y_j^{(t)}=e^{-\lambda t}Y_j(t)$. Define the pathway-specific rescaled clone-size point processes $\mathcal P_{12,t}=\sum_i\delta_{x_i^{(t)}}$ and $\mathcal P_{21,t}=\sum_j\delta_{y_j^{(t)}}$, and let $\mathcal P_t=\mathcal P_{12,t}+\mathcal P_{21,t}$. By Lemma \ref{lem:one_step_clone_process} and the conditional independence between $\mathcal P_{12,t}$ and $\mathcal P_{21,t}$, $\left(\mathcal P_{12,t},\mathcal P_{21,t}\right)\Rightarrow\left(\mathcal P_{12},\mathcal P_{21}\right)$ in the product vague topology on \(\mathcal M_p((0,\infty))^2\). Since addition is continuous under the vague topology, the continuous mapping theorem yields $\mathcal P_t \Rightarrow\mathcal P:=\mathcal P_{12}+\mathcal P_{21}$ in \(\mathcal M_p((0,\infty))\).

For $0<\varepsilon<M<\infty$, define \(S_t^{(\varepsilon,M)}=\int_{(\varepsilon,M]}x\mathcal P_t (dx)\), $T_t^{(\varepsilon,M)}=\int_{(\varepsilon,M]}x^2\mathcal P_t (dx)$, and \(R_t^{(\varepsilon,M)}=T_t^{(\varepsilon,M)}/(S_t^{(\varepsilon,M)})^2\) if \(S_t^{(\varepsilon,M)}>0\), and \(R_t^{(\varepsilon,M)}=0\) otherwise. Define the corresponding limiting quantities analogously. Because the limiting process has no atoms at $\varepsilon$ or $M$ and $\mathcal P_t \Rightarrow\mathcal P$, the continuous mapping theorem gives
\begin{equation}\label{eqn:vague convergence}
    R_t^{(\varepsilon,M)}\Rightarrow R^{(\varepsilon,M)}.
\end{equation}
Observe that $\mathbb P\left(R_t^{(\varepsilon,M)} \neq  R_t^{(\varepsilon,\infty)}\right)\le \mathbb P\left(\mathcal P_t\left(\left(M,\infty\right)\right)>0\right)$. By the Poisson construction and Markov's inequality,
\[
\begin{aligned}
\mathbb E\!\left[\mathcal P_{12,t}((M,\infty))\mid V_{s_1}\right]
&\le
\mu_{1,12}V_{s_1}
\int_{-\infty}^{t}
e^{\lambda_{s_1}u}
\min\!\left\{1,\frac{e^{-\lambda u}}{M}\right\}\,du\\
&\le
\mu_{1,12}V_{s_1}
\left(\frac{1}{\lambda_{s_1}}
+\frac{1}{\lambda-\lambda_{s_1}}\right)
M^{-\lambda_{s_1}/\lambda}.
\end{aligned}
\]
An analogous bound holds for \(\mathcal P_{21,t}\).
Both bounds are uniform in \(t\) and converge to zero almost surely
as \(M\to\infty\), since \(V_{s_1},V_{s_2}<\infty\) almost surely.
Bounding each conditional tail probability by the minimum of one
and its corresponding expected count, the union bound and dominated
convergence therefore give
\[
\lim_{M\to\infty}\sup_t
\mathbb P\!\left(\mathcal P_t((M,\infty))>0\right)=0.
\]
Hence, we have
\begin{equation}\label{eqn: M_limit_finite_process}
    \lim\limits_{M\rightarrow \infty} \sup\limits_{t}\mathbb P\left(R_t^{(\varepsilon,M)} \neq  R_t^{(\varepsilon,\infty)}\right)=0.
\end{equation}
For the limiting process, there are almost surely only finitely many points above $\varepsilon$. Hence
\begin{equation}\label{eqn: M_limit_limiting_process}
    R^{(\epsilon,M)}\rightarrow R^{(\epsilon,\infty)} \text{ almost surely as }M\rightarrow \infty.
\end{equation}
By \eqref{eqn:vague convergence}, \eqref{eqn: M_limit_finite_process}, and \eqref{eqn: M_limit_limiting_process}, we have, for any $\varepsilon>0$,
\begin{equation*}
    R_t^{(\varepsilon,\infty)}\Rightarrow R^{(\varepsilon,\infty)}.
\end{equation*}
Since these random variables lie in $[0,1]$, we have $\mathbb E\left[R_t^{(\varepsilon,\infty)}\right]\rightarrow E\left[R^{(\varepsilon,\infty)}\right]$.

It remains to remove the lower truncation at \(\varepsilon\). Define $S_t^{(0,\varepsilon)}=\int_{(0,\varepsilon)}x\,\mathcal P_t(dx)$ and $T_t^{(0,\varepsilon)}=\int_{(0,\varepsilon)}x^2\,\mathcal P_t(dx)$. Define \(S^{(0,\varepsilon)}\) and \(T^{(0,\varepsilon)}\) analogously for the limiting point process \(\mathcal P\). For \(\varepsilon>0\), we have the exact decompositions $S_t=S_t^{(0,\varepsilon)}+S_t^{(\varepsilon,\infty)}$ and $T_t=T_t^{(0,\varepsilon)}+T_t^{(\varepsilon,\infty)}$.

We first establish that the total mass contributed by clones of size at most \(\varepsilon\) is negligible uniformly in \(t\). Conditional on \(V_{s_1}\), the mutation events producing malignant clones along the \((1,2)\) pathway form a Poisson process with mutation intensity \(\mu_{1,12}V_{s_1}e^{\lambda_{s_1}u}\,du\). Since a malignant clone born at time \(u\) has expected rescaled size \(e^{-\lambda u}\), for \(\theta>0\),
\begin{align*}
\mathbb E\left[
\int_{(0,\infty)}
\left(1-e^{-\theta x}\right)
\mathcal P_{12,t}(dx)
\,\middle|\,
V_{s_1}
\right]
& \leq
\mu_{1,12}V_{s_1}
\int_{-\infty}^{t}
e^{\lambda_{s_1}u}
\min\left\{
1,\theta e^{-\lambda u}
\right\}du
\\
& \leq
\mu_{1,12}V_{s_1}
\int_{-\infty}^{\infty}
e^{\lambda_{s_1}u}
\min\left\{
1,\theta e^{-\lambda u}
\right\}du,
\end{align*}
where the first inequality is due to $1-e^{-\theta x}\le \min\left\{
1,\theta x \right\}$. Splitting the last integral at $u=\frac{\log\theta}{\lambda}$ gives
\begin{align*}
\int_{-\infty}^{\infty}
e^{\lambda_{s_1}u}
\min\left\{
1,\theta e^{-\lambda u}
\right\}du
&=
\int_{-\infty}^{(\log\theta)/\lambda}
e^{\lambda_{s_1}u}\,du
+
\theta
\int_{(\log\theta)/\lambda}^{\infty}
e^{-(\lambda-\lambda_{s_1})u}\,du
\\
&=
\left(
\frac{1}{\lambda_{s_1}}
+
\frac{1}{\lambda-\lambda_{s_1}}
\right)
\theta^{\lambda_{s_1}/\lambda}.
\end{align*}
Consequently,
\begin{align}
&\mathbb E\left[
\int_{(0,\infty)}
\left(1-e^{-\theta x}\right)
\mathcal P_{12,t}(dx)
\,\middle|\,
V_{s_1}
\right] \leq
\mu_{1,12}V_{s_1}
\left(
\frac{1}{\lambda_{s_1}}
+
\frac{1}{\lambda-\lambda_{s_1}}
\right)
\theta^{\lambda_{s_1}/\lambda}.
\label{eq:finite-time-laplace-bound-12}
\end{align}
Similarly,
\begin{align}
\mathbb E\left[
\int_{(0,\infty)}
\left(1-e^{-\theta x}\right)
\mathcal P_{21,t}(dx)
\,\middle|\,
V_{s_2}
\right]
\leq
\mu_{2,21}V_{s_2}
\left(
\frac{1}{\lambda_{s_2}}
+
\frac{1}{\lambda-\lambda_{s_2}}
\right)
\theta^{\lambda_{s_2}/\lambda}.
\label{eq:finite-time-laplace-bound-21}
\end{align}

For \(0<x\leq\varepsilon\), concavity of \(1-e^{-z}\) on \([0,1]\) gives $1-e^{-x/\varepsilon}\geq(1-e^{-1})\frac{x}{\varepsilon}$. Equivalently, $x\leq\frac{\varepsilon}{1-e^{-1}}\left(1-e^{-x/\varepsilon}\right)$. Using \eqref{eq:finite-time-laplace-bound-12} with
\(\theta=\varepsilon^{-1}\), we obtain
\begin{align*}
\mathbb E\left[
\int_{(0,\varepsilon]}
x\,\mathcal P_{12,t}(dx)
\,\middle|\,
V_{s_1}
\right]
& \leq
\frac{\varepsilon}{1-e^{-1}}
\mathbb E\left[
\int_{(0,\infty)}
\left(1-e^{-x/\varepsilon}\right)
\mathcal P_{12,t}(dx)
\,\middle|\,
V_{s_1}
\right]
\\
&\leq
\frac{\mu_{1,12}V_{s_1}}{1-e^{-1}}
\left(
\frac{1}{\lambda_{s_1}}
+
\frac{1}{\lambda-\lambda_{s_1}}
\right)
\varepsilon^{1-\lambda_{s_1}/\lambda}.
\end{align*}
Similarly,
\begin{align*}
\mathbb E\left[
\int_{(0,\varepsilon]}
x\,\mathcal P_{21,t}(dx)
\,\middle|\,
V_{s_2}
\right]
\leq
\frac{\mu_{2,21}V_{s_2}}{1-e^{-1}}
\left(
\frac{1}{\lambda_{s_2}}
+
\frac{1}{\lambda-\lambda_{s_2}}
\right)
\varepsilon^{1-\lambda_{s_2}/\lambda}.
\end{align*}

For every \(\gamma>0\), the union bound and conditional Markov
inequality therefore give
\begin{align*}
\mathbb P\left(
S_t^{(0,\varepsilon)}>\gamma
\right)
&\leq
\mathbb P\left(
\int_{(0,\varepsilon]}
x\,\mathcal P_{12,t}(dx)
>
\frac{\gamma}{2}
\right)
+
\mathbb P\left(
\int_{(0,\varepsilon]}
x\,\mathcal P_{21,t}(dx)
>
\frac{\gamma}{2}
\right)
\\
&\leq
\mathbb E\left[
1\wedge
\frac{
2\mu_{1,12}V_{s_1}
}{
\gamma(1-e^{-1})
}
\left(
\frac{1}{\lambda_{s_1}}
+
\frac{1}{\lambda-\lambda_{s_1}}
\right)
\varepsilon^{1-\lambda_{s_1}/\lambda}
\right]
\\
&\quad+
\mathbb E\left[
1\wedge
\frac{
2\mu_{2,21}V_{s_2}
}{
\gamma(1-e^{-1})
}
\left(
\frac{1}{\lambda_{s_2}}
+
\frac{1}{\lambda-\lambda_{s_2}}
\right)
\varepsilon^{1-\lambda_{s_2}/\lambda}
\right].
\end{align*}
Because $\lambda_{s_1}<\lambda$, $\lambda_{s_2}<\lambda$, and \(V_{s_1},V_{s_2}<\infty\) almost surely, dominated convergence
implies
\begin{equation}
\label{eq:uniform-small-clone-mass}
\lim_{\varepsilon\downarrow0}
\sup_t
\mathbb P\left(
S_t^{(0,\varepsilon)}>\gamma
\right)
=0.
\end{equation}

We now compare \(R_t\) with \(R_t^{(\varepsilon,\infty)}\). Fix \(\eta>0\) and \(0<\gamma<\eta/2\). On the event $\left\{S_t>\eta,S_t^{(0,\varepsilon)}\leq\gamma\right\}$, we have $S_t^{(\varepsilon,\infty)}=S_t-S_t^{(0,\varepsilon)}>\eta-\gamma>\frac{\eta}{2}$. Moreover, every clone counted in \(T_t^{(0,\varepsilon)}\) has size
at most \(\varepsilon\), so $T_t^{(0,\varepsilon)}\leq\varepsilon S_t^{(0,\varepsilon)}$. Using $T_t^{(\varepsilon,\infty)}\leq\left(S_t^{(\varepsilon,\infty)}\right)^2$, we have
\begin{align*}
\left|
R_t-R_t^{(\varepsilon,\infty)}
\right|
&\leq
\frac{
T_t^{(0,\varepsilon)}
}{
S_t^2
}
+
T_t^{(\varepsilon,\infty)}
\left[
\frac{1}{
\left(
S_t^{(\varepsilon,\infty)}
\right)^2
}
-
\frac{1}{S_t^2}
\right]
\\
&\leq
\frac{
\varepsilon S_t^{(0,\varepsilon)}
}{
\eta^2
}
+
1-
\left(
\frac{
S_t^{(\varepsilon,\infty)}
}{
S_t
}
\right)^2
\\
&=
\frac{
\varepsilon S_t^{(0,\varepsilon)}
}{
\eta^2
}
+
1-
\left(
1-
\frac{
S_t^{(0,\varepsilon)}
}{
S_t
}
\right)^2
\\
&\leq
\left(
\frac{\varepsilon}{\eta^2}
+
\frac{2}{\eta}
\right)
S_t^{(0,\varepsilon)}.
\end{align*}
Therefore, on this event, $\left|R_t-R_t^{(\varepsilon,\infty)}\right|\leq\left(\frac{\varepsilon}{\eta^2}+\frac{2}{\eta}\right)\gamma$. Since both Simpson indices lie in \([0,1]\), it follows that
\begin{align}
\mathbb E\left[
\left|
R_t-R_t^{(\varepsilon,\infty)}
\right|
\right]
\leq
\mathbb P(S_t\leq\eta)
+
\mathbb P\left(
S_t^{(0,\varepsilon)}>\gamma
\right)
+
\left(
\frac{\varepsilon}{\eta^2}
+
\frac{2}{\eta}
\right)\gamma.
\label{eq:finite-time-lower-truncation}
\end{align}
Taking \(\limsup_{t\to\infty}\) in
\eqref{eq:finite-time-lower-truncation}, then letting
\(\varepsilon\downarrow0\), and using
\eqref{eq:uniform-small-clone-mass}, gives
\begin{align*}
\limsup_{\varepsilon\downarrow0}
\limsup_{t\to\infty}
\mathbb E\left[
\left|
R_t-R_t^{(\varepsilon,\infty)}
\right|
\right]
\leq
\limsup_{t\to\infty}
\mathbb P(S_t\leq\eta)
+
\frac{2\gamma}{\eta}.
\end{align*}
Letting first \(\gamma\downarrow0\) and then \(\eta\downarrow0\) yields
\begin{equation}
\label{eq:finite-time-lower-truncation-vanishes}
\lim_{\varepsilon\downarrow0}
\limsup_{t\to\infty}
\mathbb E\left[
\left|
R_t-R_t^{(\varepsilon,\infty)}
\right|
\right]
=0.
\end{equation}

For the limiting process, finiteness of \(S\) implies $S^{(0,\varepsilon)}
\rightarrow 0$ almost surely as $\varepsilon\rightarrow 0$. Moreover, $T^{(0,\varepsilon)}\leq\varepsilon S^{(0,\varepsilon)}\rightarrow0$ almost surely. Consequently, $S^{(\varepsilon,\infty)}\rightarrow S$ and $T^{(\varepsilon,\infty)}\rightarrow T$ almost surely. Since \(S>0\) almost surely, $R^{(\varepsilon,\infty)}\rightarrow R$ almost surely. Because $R^{(\varepsilon,\infty)}$ and $R$ lie in $[0,1]$, dominated convergence gives
\begin{equation}
\label{eq:limiting-lower-truncation-vanishes}
\lim_{\varepsilon\downarrow0}
\mathbb E\left[
\left|
R^{(\varepsilon,\infty)}-R
\right|
\right]
=0.
\end{equation}

Finally,
\begin{align*}
\left|
\mathbb E[R_t]-\mathbb E[R]
\right|
\leq
\mathbb E\left[
\left|
R_t-R_t^{(\varepsilon,\infty)}
\right|
\right]
+
\left|
\mathbb E\left[
R_t^{(\varepsilon,\infty)}
\right]
-
\mathbb E\left[
R^{(\varepsilon,\infty)}
\right]
\right|
+
\mathbb E\left[
\left|
R^{(\varepsilon,\infty)}-R
\right|
\right].
\end{align*}
For fixed \(\varepsilon>0\), the middle term converges to zero as
\(t\to\infty\). Therefore,
\begin{align*}
\limsup_{t\to\infty}
\left|
\mathbb E[R_t]-\mathbb E[R]
\right|
\leq
\limsup_{t\to\infty}
\mathbb E\left[
\left|
R_t-R_t^{(\varepsilon,\infty)}
\right|
\right]
+
\mathbb E\left[
\left|
R^{(\varepsilon,\infty)}-R
\right|
\right].
\end{align*}
Letting \(\varepsilon\downarrow0\) and using \eqref{eq:finite-time-lower-truncation-vanishes} and \eqref{eq:limiting-lower-truncation-vanishes}, we conclude that $\lim_{t\to\infty}\mathbb E[R_t]=\mathbb E[R]$.
\end{proof}

\subsection*{Acknowledgements}
The authors would like to thank Dr. Xinyun Chen for helpful comments on the draft. The work of ZW was supported in part by the Key Program of the National Natural Science Foundation of China (NSFC) under Grant No. 72495131 and the Guangdong (China) Provincial Key Laboratory of Mathematical Foundations for Artificial Intelligence [2023B1212010001].

\subsection*{Declaration of generative AI in the manuscript preparation process}
During the preparation of this work the authors used ChatGPT in order to refine the writing and enhance the linguistic quality of the manuscript. After using this tool/service, the authors reviewed and edited the content as needed and take full responsibility for the content of the manuscript.

\newpage

\clearpage
\begin{singlespace}

\bibliography{bibliography}

\begin{thebibliography}{}

\bibitem[Ahmed et~al., 2026]{ahmed2026site}
Ahmed, A., Foo, J., Gunnarsson, E., and Leder, K. (2026).
\newblock The site frequency spectrum in an exponentially-growing population with selection.
\newblock {\em arXiv preprint arXiv:2607.16479}.

\bibitem[Armitage and Doll, 1954]{armitage1954age}
Armitage, P. and Doll, R. (1954).
\newblock The age distribution of cancer and a multi-stage theory of carcinogenesis.
\newblock {\em British journal of cancer}, 8:1--12.

\bibitem[Avanzini and Antal, 2019]{avanzini2019cancer}
Avanzini, S. and Antal, T. (2019).
\newblock Cancer recurrence times from a branching process model.
\newblock {\em PLoS computational biology}, 15(11):e1007423.

\bibitem[Beerenwinkel et~al., 2007]{beerenwinkel2007genetic}
Beerenwinkel, N., Antal, T., Dingli, D., Traulsen, A., Kinzler, K.~W., Velculescu, V.~E., Vogelstein, B., and Nowak, M.~A. (2007).
\newblock Genetic progression and the waiting time to cancer.
\newblock {\em PLoS computational biology}, 3(11):e225.

\bibitem[Bozic et~al., 2010]{bozic2010accumulation}
Bozic, I., Antal, T., Ohtsuki, H., Carter, H., Kim, D., Chen, S., Karchin, R., Kinzler, K.~W., Vogelstein, B., and Nowak, M.~A. (2010).
\newblock Accumulation of driver and passenger mutations during tumor progression.
\newblock {\em Proceedings of the National Academy of Sciences}, 107(43):18545--18550.

\bibitem[Cheek and Antal, 2018]{cheek2018mutation}
Cheek, D. and Antal, T. (2018).
\newblock Mutation frequencies in a birth--death branching process.
\newblock {\em The Annals of Applied Probability}, 28(6):3922--3947.

\bibitem[Cheek and Antal, 2020]{cheek2020genetic}
Cheek, D. and Antal, T. (2020).
\newblock Genetic composition of an exponentially growing cell population.
\newblock {\em Stochastic Processes and their Applications}, 130(11):6580--6624.

\bibitem[Dagogo-Jack and Shaw, 2018]{dagogo2018tumour}
Dagogo-Jack, I. and Shaw, A.~T. (2018).
\newblock Tumour heterogeneity and resistance to cancer therapies.
\newblock {\em Nature reviews Clinical oncology}, 15(2):81--94.

\bibitem[Davis et~al., 2017]{davis2017tumor}
Davis, A., Gao, R., and Navin, N. (2017).
\newblock Tumor evolution: Linear, branching, neutral or punctuated?
\newblock {\em Biochimica et Biophysica Acta (BBA)-Reviews on Cancer}, 1867(2):151--161.

\bibitem[Dewanji et~al., 2011]{dewanji2011number}
Dewanji, A., Jeon, J., Meza, R., and Luebeck, E.~G. (2011).
\newblock Number and size distribution of colorectal adenomas under the multistage clonal expansion model of cancer.
\newblock {\em PLoS Computational Biology}, 7(10):e1002213.

\bibitem[Durrett et~al., 2011]{durrett2011intratumor}
Durrett, R., Foo, J., Leder, K., Mayberry, J., and Michor, F. (2011).
\newblock Intratumor heterogeneity in evolutionary models of tumor progression.
\newblock {\em Genetics}, 188(2):461--477.

\bibitem[Durrett and Moseley, 2010]{durrett2010evolution}
Durrett, R. and Moseley, S. (2010).
\newblock Evolution of resistance and progression to disease during clonal expansion of cancer.
\newblock {\em Theoretical population biology}, 77(1):42--48.

\bibitem[Durrett et~al., 2009]{durrett2009waiting}
Durrett, R., Schmidt, D., and Schweinsberg, J. (2009).
\newblock A waiting time problem arising from the study of multi-stage carcinogenesis.
\newblock {\em The Annals of Applied Probability}, 19(2):676--718.

\bibitem[Gerlinger et~al., 2012]{gerlinger2012intratumor}
Gerlinger, M., Rowan, A.~J., Horswell, S., Larkin, J., Endesfelder, D., Gronroos, E., Martinez, P., Matthews, N., Stewart, A., Tarpey, P., et~al. (2012).
\newblock Intratumor heterogeneity and branched evolution revealed by multiregion sequencing.
\newblock {\em New England journal of medicine}, 366(10):883--892.

\bibitem[Gunnarsson et~al., 2021]{gunnarsson2021exact}
Gunnarsson, E.~B., Leder, K., and Foo, J. (2021).
\newblock Exact site frequency spectra of neutrally evolving tumors: A transition between power laws reveals a signature of cell viability.
\newblock {\em Theoretical Population Biology}, 142:67--90.

\bibitem[Haeno et~al., 2007]{haeno2007evolution}
Haeno, H., Iwasa, Y., and Michor, F. (2007).
\newblock The evolution of two mutations during clonal expansion.
\newblock {\em Genetics}, 177(4):2209--2221.

\bibitem[Hanagal et~al., 2022]{hanagal2022large}
Hanagal, P., Leder, K., and Wang, Z. (2022).
\newblock Large deviations of cancer recurrence timing.
\newblock {\em Stochastic Processes and their Applications}, 147:1--50.

\bibitem[Iwasa et~al., 2006]{iwasa2006evolution}
Iwasa, Y., Nowak, M.~A., and Michor, F. (2006).
\newblock Evolution of resistance during clonal expansion.
\newblock {\em Genetics}, 172(4):2557--2566.

\bibitem[Kent and Green, 2017]{kent2017order}
Kent, D.~G. and Green, A.~R. (2017).
\newblock Order matters: the order of somatic mutations influences cancer evolution.
\newblock {\em Cold Spring Harbor perspectives in medicine}, 7(4):a027060.

\bibitem[Knudson~Jr, 1971]{knudson1971mutation}
Knudson~Jr, A.~G. (1971).
\newblock Mutation and cancer: statistical study of retinoblastoma.
\newblock {\em Proceedings of the National Academy of Sciences}, 68(4):820--823.

\bibitem[Last and Penrose, 2017]{last2017lectures}
Last, G. and Penrose, M. (2017).
\newblock {\em Lectures on the Poisson Process}, volume~7 of {\em Institute of Mathematical Statistics Textbooks}.
\newblock Cambridge University Press.

\bibitem[Leder et~al., 2024]{leder2024parameter}
Leder, K., Sun, R., Wang, Z., and Zhang, X. (2024).
\newblock Parameter estimation from single patient, single time-point sequencing data of recurrent tumors.
\newblock {\em Journal of Mathematical Biology}, 89(5):51.

\bibitem[Leder and Wang, 2026]{leder2026clonal}
Leder, K. and Wang, Z. (2026).
\newblock Clonal diversity at early cancer recurrence.
\newblock {\em Bulletin of Mathematical Biology}, 88(5):75.

\bibitem[Leder et~al., 2025]{leder2025parameter}
Leder, K., Wang, Z., and Zhang, X. (2025).
\newblock Parameter estimation in recurrent tumor evolution with finite carrying capacity.
\newblock {\em arXiv preprint arXiv:2510.01078}.

\bibitem[Levine et~al., 2019]{levine2019roles}
Levine, A.~J., Jenkins, N.~A., and Copeland, N.~G. (2019).
\newblock The roles of initiating truncal mutations in human cancers: the order of mutations and tumor cell type matters.
\newblock {\em Cancer cell}, 35(1):10--15.

\bibitem[Li et~al., 2023a]{li2023comparison}
Li, A., Kibby, D., and Foo, J. (2023a).
\newblock A comparison of mutation and amplification-driven resistance mechanisms and their impacts on tumor recurrence.
\newblock {\em Journal of Mathematical Biology}, 87(4):59.

\bibitem[Li et~al., 2023b]{li2023mathematical}
Li, L., Hu, Y., Xu, Y., and Tang, S. (2023b).
\newblock Mathematical modeling the order of driver gene mutations in colorectal cancer.
\newblock {\em PLOS Computational Biology}, 19(6):e1011225.

\bibitem[Mahieu et~al., 2024]{mahieu2024oraov1}
Mahieu, C.~I., Mancini, A.~G., Vikram, E.~P., Planells-Palop, V., Joseph, N.~M., and Tward, A.~D. (2024).
\newblock Oraov1, ccnd1, and mir548k are the driver oncogenes of the 11q13 amplicon in squamous cell carcinoma.
\newblock {\em Molecular Cancer Research}, 22(2):152--168.

\bibitem[Maley et~al., 2006]{maley2006genetic}
Maley, C.~C., Galipeau, P.~C., Finley, J.~C., Wongsurawat, V.~J., Li, X., Sanchez, C.~A., Paulson, T.~G., Blount, P.~L., Risques, R.-A., Rabinovitch, P.~S., et~al. (2006).
\newblock Genetic clonal diversity predicts progression to esophageal adenocarcinoma.
\newblock {\em Nature genetics}, 38(4):468--473.

\bibitem[McDonald and Kimmel, 2015]{mcdonald2015multitype}
McDonald, T.~O. and Kimmel, M. (2015).
\newblock A multitype infinite-allele branching process with applications to cancer evolution.
\newblock {\em Journal of Applied Probability}, 52(3):864--876.

\bibitem[Nicholson and Antal, 2016]{nicholson2016universal}
Nicholson, M.~D. and Antal, T. (2016).
\newblock Universal asymptotic clone size distribution for general population growth.
\newblock {\em Bulletin of Mathematical Biology}, 78(11):2243--2276.

\bibitem[Nicholson and Antal, 2019]{nicholson2019competing}
Nicholson, M.~D. and Antal, T. (2019).
\newblock Competing evolutionary paths in growing populations with applications to multidrug resistance.
\newblock {\em PLoS computational biology}, 15(4):e1006866.

\bibitem[Nicholson et~al., 2023]{nicholson2023sequential}
Nicholson, M.~D., Cheek, D., and Antal, T. (2023).
\newblock Sequential mutations in exponentially growing populations.
\newblock {\em PLoS computational biology}, 19(7):e1011289.

\bibitem[Norrie et~al., 2021]{norrie2021retinoblastoma}
Norrie, J.~L., Nityanandam, A., Lai, K., Chen, X., Wilson, M., Stewart, E., Griffiths, L., Jin, H., Wu, G., Orr, B., et~al. (2021).
\newblock Retinoblastoma from human stem cell-derived retinal organoids.
\newblock {\em Nature Communications}, 12(1):4535.

\bibitem[Nowell, 1976]{nowell1976clonal}
Nowell, P.~C. (1976).
\newblock The clonal evolution of tumor cell populations: Acquired genetic lability permits stepwise selection of variant sublines and underlies tumor progression.
\newblock {\em Science}, 194(4260):23--28.

\bibitem[Ortmann et~al., 2015]{ortmann2015effect}
Ortmann, C.~A., Kent, D.~G., Nangalia, J., Silber, Y., Wedge, D.~C., Grinfeld, J., Baxter, E.~J., Massie, C.~E., Papaemmanuil, E., Menon, S., et~al. (2015).
\newblock Effect of mutation order on myeloproliferative neoplasms.
\newblock {\em New England Journal of Medicine}, 372(7):601--612.

\bibitem[Paterson et~al., 2020]{paterson2020mathematical}
Paterson, C., Clevers, H., and Bozic, I. (2020).
\newblock Mathematical model of colorectal cancer initiation.
\newblock {\em Proceedings of the National Academy of Sciences}, 117(34):20681--20688.

\bibitem[Smith and Haigh, 1974]{smith1974hitch}
Smith, J.~M. and Haigh, J. (1974).
\newblock The hitch-hiking effect of a favourable gene.
\newblock {\em Genetics Research}, 23(1):23--35.

\bibitem[Stein and Werner, 2025]{stein2025patterns}
Stein, A. and Werner, B. (2025).
\newblock On the patterns of genetic intra-tumor heterogeneity before and after treatment.
\newblock {\em Genetics}, 230(4):iyaf101.

\bibitem[Storey et~al., 2017]{storey2017spatial}
Storey, K., Ryser, M.~D., Leder, K., and Foo, J. (2017).
\newblock Spatial measures of genetic heterogeneity during carcinogenesis.
\newblock {\em Bulletin of mathematical biology}, 79(2):237--276.

\bibitem[Teimouri and Kolomeisky, 2021]{teimouri2021temporal}
Teimouri, H. and Kolomeisky, A.~B. (2021).
\newblock Temporal order of mutations influences cancer initiation dynamics.
\newblock {\em Physical Biology}, 18(5):056002.

\bibitem[Vogelstein and Kinzler, 2004]{vogelstein2004cancer}
Vogelstein, B. and Kinzler, K.~W. (2004).
\newblock Cancer genes and the pathways they control.
\newblock {\em Nature medicine}, 10(8):789--799.

\bibitem[Wang et~al., 2024]{wang2024order}
Wang, Y., Shtylla, B., and Chou, T. (2024).
\newblock Order-of-mutation effects on cancer progression: models for myeloproliferative neoplasm.
\newblock {\em Bulletin of Mathematical Biology}, 86(3):32.

\bibitem[Zhang and Bozic, 2024]{zhang2024accumulation}
Zhang, R. and Bozic, I. (2024).
\newblock Accumulation of oncogenic mutations during progression from healthy tissue to cancer.
\newblock {\em Bulletin of Mathematical Biology}, 86(12):1--33.

\bibitem[Zhang et~al., 2023]{zhang2023waiting}
Zhang, R., Ukogu, O.~A., and Bozic, I. (2023).
\newblock Waiting times in a branching process model of colorectal cancer initiation.
\newblock {\em Theoretical Population Biology}, 151:44--63.

\end{thebibliography}

\end{singlespace}

\end{document}